\documentclass[twoside,11pt]{article}
\usepackage{amssymb,amsmath,amsthm,amsfonts, epsfig}           
\usepackage{mathtools} 
\usepackage{hyperref}
\usepackage{eucal}
\usepackage{a4wide}
\usepackage[svgnames]{xcolor}
\usepackage{mathptmx} 
\usepackage{graphicx}
\usepackage{tikz-cd}
 \usepackage{subfigure}
 \usepackage{fancyhdr}
 \usepackage{MnSymbol}

\title{Hamiltonisation, measure preservation and first integrals of the multi-dimensional rubber Routh sphere\footnote{This research was made possible by a Georg Forster Experienced Researcher Fellowship  from the 
Alexander von Humboldt Foundation that funded a research stay of the author at TU Berlin.}}

\author{ Luis C.~Garc\'ia-Naranjo }

\numberwithin{equation}{section}
\numberwithin{table}{section}
\numberwithin{figure}{section}

\hypersetup{colorlinks=true,linkcolor=violet,citecolor=purple}

\newtheorem{theorem}{Theorem}[section]
\newtheorem{lemma}[theorem]{Lemma}
\newtheorem{proposition}[theorem]{Proposition}
\newtheorem{corollary}[theorem]{Corollary}

{\theoremstyle{definition}
\newtheorem{definition}[theorem]{Definition}
\newtheorem{remark}[theorem]{Remark}
\newtheorem*{remarks*}{Remarks}

}

\newtheorem{innercustomgeneric}{\customgenericname}
\providecommand{\customgenericname}{}
\newcommand{\newcustomtheorem}[2]{%
  \newenvironment{#1}[1]
  {%
   \renewcommand\customgenericname{#2}%
   \renewcommand\theinnercustomgeneric{##1}%
   \innercustomgeneric
  }
  {\endinnercustomgeneric}
}

\newcustomtheorem{customhyp}{}

\newcommand{\defn}[1]{{\bfseries\itshape{#1}}}

\def\headcolour{\color{Grey}}
\def\restr#1{\,\vrule height1.2ex width.4pt
  depth0.8ex\lower0.4ex\hbox{\scriptsize $\,#1$}}

\newcommand{\R}{\mathbb{R}}

\newcommand{\I}{\mathbb{I}}
\newcommand{\J}{\mathbb{J}}
\newcommand{\g}{\mathfrak{g}}
\newcommand{\so}{\mathfrak{so}}

\newcommand{\SO}{\mathrm{SO}}
\newcommand{\SE}{\mathrm{SE}}
\newcommand{\Ad}{\mathrm{Ad}}
\newcommand{\Ss}{\mathrm{S}}

\newcommand{\se}{\frak{se}}
\newcommand{\tr}{\mathop\mathrm{tr}\nolimits}

\begin{document}

\maketitle

\begin{abstract}
We consider the multi-dimensional generalisation of the problem of a sphere, with  axi-symmetric mass distribution, that rolls without slipping
or spinning over a plane. Using recent results from Garc\'ia-Naranjo~\cite{LGN18} and
Garc\'ia-Naranjo and Marrero~\cite{GNMarr2018}, we show that the reduced equations of motion possess an invariant measure and may be represented in
Hamiltonian form by Chaplygin's reducing multiplier method. 
We also prove a general result on the existence of first integrals for certain
  Hamiltonisable Chaplygin systems with 
internal symmetries that is used  to determine  
conserved quantities of the problem.
  \end{abstract}

{\small
\tableofcontents
}

\section{Introduction}

An important contribution of S. A. Chaplygin to the field of nonholonomic systems was the introduction of the
so-called {\em Chaplygin's reducing multiplier method}~\cite{ChapRedMult}. 
It is concerned with a certain class
of nonholonomic systems
with symmetry, commonly referred to today  as {\em nonholonomic Chaplygin systems}, whose reduced 
equations of motion have the form of a classical mechanical system subjected to extra gyroscopic forces.
Chaplygin's method consists of searching for a {\em time reparametrisation} hoping that in the 
new time variable, and after a momentum rescaling, the extra forces vanish and the resulting system is 
Hamiltonian. If successful, this process is often 
referred to as \defn{Chaplygin Hamiltonisation}. It is also common to 
say that the Chaplygin system at hand, in the original time variable, is \defn{conformally Hamiltonian}.
The subclass of Chaplygin systems allowing a Chaplygin Hamiltonisation is
quite remarkable, and substantial effort has been devoted to their study and characterisation (see e.g.~\cite{ChapRedMult,Iliev1985,
Stanchenko,BorMamHam, FedJov, EhlersKoiller, BolsBorMam2015, Fernandez, LGN18} and references therein). The purpose
of this paper is to provide a new non-trivial example within this category, 
and to prove a considerably general Noether-type of result for these systems, 
which links their internal symmetries with first integrals, and which we apply to our example.

\subsection*{$\phi$-simple Chaplygin systems: Hamiltonisation and first integrals.}

Our approach to Chaplygin Hamiltonisation relies on the notion of {\em $\phi$-simple Chaplygin systems} 
 introduced recently in Garc\'ia-Naranjo and Marrero~\cite{GNMarr2018}. These
 systems form an exceptional  subclass of nonholonomic Chaplygin systems that always 
 possess an invariant measure and
allow a Chaplygin Hamiltonisation.

 The definition of $\phi$-simple Chaplygin systems relies  on 
  a  certain tensor field  $\mathcal{T}$ of type $(1,2)$ defined on the {\em shape space}\footnote{the
 shape space is the quotient manifold $S=Q/G$ where $Q$ is the configuration 
 manifold of the system and $G$ is the underlying symmetry group of the Chaplygin system.}
$S$ of the system, which measures the interplay between the kinetic energy and the non-integrability
of the constraint distribution. This tensor field already  appears in the works of  Koiller~\cite{Koi} 
and Cantrijn et al~\cite{CaCoLeMa}, and, following the terminology of~\cite{LGN18, GNMarr2018}, will be
called the {\em gyroscopic tensor}. A Chaplygin system is said to be \defn{$\phi$-simple} if there exists
a   function $\phi \in C^\infty(S)$ such that the gyroscopic 
 tensor ${\mathcal T}$ satisfies\footnote{throughout the paper
 we denote by $Y[f]$ the action of the vector field $Y$ on the 
scalar function $f$.}
\begin{equation}\label{Ham-condition-intro}
{\mathcal T}(Y, Z) =Z[\phi]Y - Y[\phi]Z, 
\end{equation}
for any two vector fields $ Y, Z$ on $S$. The above  condition was obtained as 
the coordinate-free formulation
of the recent results on Chaplygin Hamiltonisation given recently by the author~\cite{LGN18}.
  It is shown in~\cite{GNMarr2018} that the condition to be $\phi$-simple
 is equivalent to 
the verification of certain sufficient conditions for 
Chaplygin Hamiltonisation given previously by Stanchenko~\cite{Stanchenko} and
 Cantrijn et al~\cite{CaCoLeMa}. 
The advantage of the formulation in~\cite{LGN18} and \cite{GNMarr2018} with respect to these
references is that condition~\eqref{Ham-condition-intro} can 
be systematically examined in concrete examples.  

The statement that a $\phi$-simple Chaplygin system allows a Chaplygin Hamiltonisation is independent
of the number of degrees of freedom of the problem, and may be interpreted as a generalisation of the 
celebrated Chaplygin's Reducing Multiplier Theorem~\cite{ChapRedMult}
 - whose applicability is restricted to systems whose
 shape space has dimension $2$ - see the discussion in~\cite{LGN18} and \cite{GNMarr2018}.

The criterion of $\phi$-simplicity has already been used in~\cite{LGN18} and \cite{GNMarr2018} 
to establish the Hamiltonisation of non-trivial examples. Among them is  the multi-dimensional Veselova 
problem whose Hamiltonisation was first proven by Fedorov and Jovanovi\'c~\cite{FedJov, FedJov2} 
by a direct application of Chaplygin's method to the reduced equations of motion. 
 In Section~\ref{S:Example} of this paper we prove that the multi-dimensional rubber Routh sphere 
(introduced below) is  also $\phi$-simple. This allows us to prove that the system allows
a Chaplygin Hamiltonisation, and to give a closed formula for its invariant measure,  without
writing the equations of motion. Our results seem to support the thesis that $\phi$-simplicity
is the relevant mechanism behind the Chaplygin Hamiltonisation of concrete examples.\footnote{this
statement is meant within the framework of Hamiltonisation of Chaplygin systems. Other examples, 
like the remarkable Hamiltonisation of the Chaplygin sphere obtained by
Borisov and Mamaev~\cite{BorMamChap}, involve a further symmetry reduction \cite{BorMamHam} 
and other geometric
mechanisms come into play to ensure that invariant first integrals descend to the quotient space
as Casimir functions~\cite{LGN-JM17}. }

In Section~\ref{S:Noether} we prove  a general result that
shows how  extra symmetries of $\phi$-simple Chaplygin systems leads to the existence
of conserved quantities (Theorem~\ref{Th:Noether}). This result is applied to find first integrals 
of the multi-dimensional rubber Routh sphere in Section~\ref{SS:Integrability}. 
This contributes to the recent
efforts to understand the mechanisms responsible for the existence of
first integrals that are linear in velocities in nonholonomic mechanics~(see e.g. \cite{Iliev1975, Fasso2, FassoJGM,BalseiroSansonetto}).

\subsection*{The multi-dimensional rubber Routh sphere}

Routh~\cite{Routh} considered the problem of a sphere, whose
distribution of mass is axially symmetric, that rolls without slipping on the plane. 
Later, Borisov and Mamaev~\cite{BorMam2008, BorMamRubber} considered the 
problem under an additional {\em rubber}\footnote{the {\em rubber} terminology 
for constraints that prohibit spinning  goes back to Ehlers et al~\cite{EhlersKoiller} and
Koiller and Ehlers~\cite{KoillerRubber}
and is now quite standard in the field.} constraint
that forbids spinning. In this paper we consider the multi-dimensional generalisation of this system.
Our terminology {\em multi-dimensional rubber Routh sphere} is supposed to 
indicate the presence of a no-spin  constraint in the word {\em rubber}, and the axi-symmetric 
assumption on the mass
distribution of the sphere with the mention of {\em Routh}'s name. A closely related problem is the multi-dimensional
rubber Chaplygin sphere considered by Jovanovi\'c~\cite{JovaRubber}.

The study of multi-dimensional systems in nonholonomic mechanics goes back to Fedorov and
Kozlov~\cite{FedKoz}, and has received wide attention as a source of interesting 
examples for integrability, Hamiltonisation and other types of dynamical 
features~\cite{ZenkovSuslov,Jov2003, FedJov,FedJov2,Jovan,Jova18,FGNS18,FGNM18,Gajic}.
Our analysis of the multi-dimensional rubber Routh sphere contributes to enlarge this family of examples.

\subsection*{Structure of the paper}

In Section~\ref{S:ChapSyst} we present a quick review of the recent 
constructions in~\cite{LGN18, GNMarr2018}. This summary includes 
the definition of the gyroscopic tensor of 
and its expression in local coordinates. We also recall the notion of $\phi$-simplicity
(described above) and, in Theorem~\ref{T:phi-simple}, we indicate its precise relationship with measure preservation and Hamiltonisation. Section~\ref{S:Noether} is completely devoted to
Theorem~\ref{Th:Noether} that relates internal symmetries of $\phi$-simple Chaplygin systems
to first integrals. Section~\ref{S:Example} is concerned with the multidimensional rubber Routh sphere.
To simplify the reading, we first treat the 3D system in subsection~\ref{SS:3D} and then proceed to the $n$D
generalisation  in
subsection~\ref{SS:nD}. The $\phi$-simplicity of the system is presented in 
Theorem~\ref{T:phiSimpleRouthnD}
and the consequential measure preservation and Hamiltonisation properties 
in Corollary~\ref{Cor:MeasHam}.
The results of  Section~\ref{S:Noether} are then applied to 
determine first integrals of the problem in subsection~\ref{SS:Integrability}. 
The paper finishes with an Appendix~\ref{app:proofs-lemmas} that contains the
proof of a technical lemma needed in the proof of  Theorem~\ref{T:phiSimpleRouthnD}.

\section{Preliminaries: a review of $\phi$-simple Chaplygin systems and their measure preservation
and Hamiltonisation properties}
\label{S:ChapSyst}

In this section we briefly recall the notion of  nonholonomic Chaplygin systems and, more specifically,
 $\phi$-simple Chaplygin systems introduced in Garc\'ia-Naranjo and Marrero~\cite{GNMarr2018}, together with their measure preservation and Hamiltonisation properties.

\subsection{Nonholonomic Chaplygin systems and the gyroscopic tensor}
For our purposes, a \defn{ nonholonomic system} is a triple  $(Q,D,L)$. Here $Q$ is an  $n$-dimensional smooth 
manifold modelling 
the configuration space of the system. $D\subset TQ$ is a vector sub-bundle whose fibres define a non-integrable 
distribution on $Q$ of constant rank  $r\geq 2$,  that models $n-r$ linear 
nonholonomic constraints as follows: a curve $q(t)$ on $Q$ is said to satisfy the constraints
if and only if $\dot q(t)\in D_{q(t)}$  for all $t$. Finally, $L:TQ\to \R$ is the Lagrangian of the system that is assumed to be of 
{\em mechanical type}, namely 
\begin{equation*}
L=K-U,
\end{equation*}
 where the kinetic energy $K$ defines a Riemannian metric $\llangle \cdot , \cdot \rrangle$ on 
 $Q$, and $U:Q\to \R$ is the potential energy.

The triple $(Q,D,L)$ contains all the information for  the evolution of the system in accordance with the
 {\em Lagrange-D'Alembert principle} of ideal constraints. The (velocity) phase space of the system is $D$ and the dynamics is
  described by the flow of a uniquely defined vector field $X_{nh}\in \frak{X}(D)$. An intrinsic definition of this vector field may
  be found, for instance, in~\cite{LeMa}.

\begin{definition}
\label{D:Chap}
The nonholonomic system $(Q,D,L)$ is said to be a \defn{Chaplygin system} if there exists an ($n-r$)-dimensional  Lie group $G$ 
acting freely and properly on $Q$ and satisfying the following properties:
\begin{enumerate}
\item $G$ acts by isometries with respect to the kinetic energy metric $\llangle \cdot , \cdot \rrangle$, and the potential energy $U$ is invariant,
\item $D$ is invariant in the sense that $D_{g\cdot q}=Tg(D_q)$ for all $g\in G$ and $q\in Q$,
\item for all $q\in Q$ the following direct sum splitting holds
\begin{equation*}
T_qQ =  \g \cdot q \oplus D_q,
\end{equation*}
where $\g$ denotes the Lie algebra of $q$, and $\g \cdot q$ is the tangent space to the orbit through $q$ at $q$.
\end{enumerate}
\end{definition}

\begin{remark} Chaplygin systems as defined above 
are also referred to in  the literature as {\em non-abelian Chaplygin systems}  \cite{Koi},
 {\em generalised Chaplgyin systems} \cite{CaCoLeMa,FedJov} or the {\em principal kinematic case of a nonholonomic
 system with symmetries}~\cite{BKMM}. 
\end{remark}

The smooth $r$-dimensional manifold $S:=Q/G$ associated to a Chaplygin system is called the \defn{shape space}.
 As a consequence of the  first and second conditions in Definition~\ref{D:Chap},  the vector field $X_{nh}$ describing the dynamics
  is equivariant (with respect to the $G$-action on $D$ defined by the restriction of the  tangent lifted action of $G$ to $TQ$)  and the system
admits a $G$-reduction. The reduced dynamics is described by the flow of the \defn{reduced  vector field} 
$\bar X_{nh}$
on the orbit space $D/G$. 
 For a Chaplygin system, the reduced phase space $D/G$ is naturally identified with the tangent bundle
 $TS$, or cotangent bundle $T^*S$, using the (reduced) Legendre transform.
The  reduced equations of motion have the form of a mechanical system on the shape space $S$ subject to gyroscopic forces. 
Geometrically, they may be written in almost symplectic form, i.e. in Hamiltonian-like form, 
${\bf i}_{\bar X_{nh}} \Omega_{nh} =dH$,
where $H$ 
is the reduced Hamiltonian (energy),  but where the non-degenerate 2-form $\Omega_{nh}$ on $T^*S$ fails, in general,  to be closed.
See section~\ref{SS:reduced-eqns} below or references~\cite{EhlersKoiller,FedJov,HochGN, GNMarr2018} for more details.

We now  recall  the definition of the {\em gyroscopic tensor}  from Garc\'ia-Naranjo and Marrero~\cite{GNMarr2018}.
To do this, we first note that the kinetic energy metric defines the orthogonal decomposition $TQ = D\oplus D^\perp$.
We shall denote by 
\begin{equation*}
\mathcal{P}: TQ\to D,
\end{equation*}
the bundle projection associated to such decomposition. Next we note that, as was first pointed out by Koiller~\cite{Koi}, the
second and  third conditions in Definition~\ref{D:Chap} imply that the fibres of $D$ are the horizontal spaces of a principal connection
on the principal bundle $\pi:Q\to Q/G=S$. As is well known, corresponding to such principal connection, there is a well defined \defn{horizontal lift} that associates  to any vector field $ Y\in \frak{X}(S)$ an equivariant vector field 
$\mbox{hor}( Y)\in \frak{X}(Q)$ taking values on $D$, and that 
is $\pi$-related to $ Y$. We are now ready to present:

\begin{definition}
\label{D:defGyrTensor}
The  \defn{gyroscopic tensor}  $\mathcal{T}$  is the $(1,2)$ skew-symmetric tensor field on 
$S$ determined by assigning to  the vector fields 
$Y, Z\in \frak{X}(S)$,  the vector field $\mathcal{T}(Y,Z)\in \frak{X}(S)$, given by
\begin{equation}\label{eq:gyroscopic-tensor}
\mathcal{T}(Y,Z)(s)=(T_q\pi) \left ( \mathcal{P} \left [ \mbox{hor}(Y) \, , \,  \mbox{hor}(Z)   \right ] (q)\right ) - [Y,Z](s),
\end{equation}
for $s \in S$, and where $q \in Q$ is any point such that $\pi(q) = s$, and where $[\cdot, \cdot ]$ denotes the Jacobi-Lie bracket
of vector fields.
\end{definition}

That $\mathcal{T}$  is a well-defined  $(1,2)$ tensor field on $S$ is shown in~\cite[Proposition 3.4]{GNMarr2018}. It is also 
shown in this reference that the gyroscopic tensor $\mathcal{T}$ coincides with other tensor fields that had been considered 
before by Koiller~\cite{Koi} and Cantrijn et al~\cite{CaCoLeMa}.

\subsection{Local expressions for the gyroscopic tensor}
\label{SS:local-expressions}


Let $s=(s^1, \dots, s^r)$ be local coordinates on the shape space $S$. Then we may write
\begin{equation*}
\mathcal{T} \left ( \frac{\partial}{\partial s^i} ,  \frac{\partial}{\partial s^j}\right )  = \sum_{k=1}^r C_{ij}^k  \frac{\partial}{\partial s^k},
\end{equation*}
for certain $s$-dependent coefficients $C_{ij}^k$, which, in view of the skew-symmetry of $\mathcal{T}$, are skew-symmetric
with respect to the lower indices, i.e. $C_{ij}^k=-C_{ji}^k$. Following the terminology introduced in 
Garc\'ia-Naranjo~\cite{LGN18}, we refer to 
$C_{ij}^k$ as the \defn{gyroscopic coefficients}.

Fix $i,j\in \{1,\dots r\}$. The gyroscopic coefficients $C_{ij}^k$, $k=1, \dots, r$, may be computed in practice by 
solving the following linear system of equations:
\begin{equation}
\label{eq:gyr-coeff}
\sum_{k=1}^r K_{kl} C_{ij}^k = \llangle [h_i,h_j], h_l \rrangle, \qquad l=1,\dots, r,
\end{equation}
where we recall that  $\llangle \cdot , \cdot \rrangle$ is the kinetic energy metric on $Q$, and we have denoted
\begin{equation}
\label{eq:local-S-metric}
h_k:= \mbox{hor}\left (  \frac{\partial}{\partial s^k} \right ), \qquad K_{kl}:= \llangle h_k , h_l \rrangle, \qquad k,l=1,\dots, r.
\end{equation}
Note that the matrix $K_{kl}$ is invertible by linear independence of $\{h_1, \dots , h_l\}$.
 That~\eqref{eq:gyr-coeff}
holds is a direct consequence of the definition~\eqref{eq:gyroscopic-tensor} of the gyroscopic tensor since
 $\left [ \frac{\partial}{\partial s^i} ,  \frac{\partial}{\partial s^j} \right ]=0$.

\subsection{The reduced equations of motion}
\label{SS:reduced-eqns}

We now present the reduced equations of motion  of  a Chaplygin system. 
 Our exposition mainly follows 
Garc\'ia-Naranjo~\cite{LGN18}.

The distribution $D$ interpreted as a principal connection on the principal bundle $\pi:Q\to S$, 
induces a Riemannian metric on $S$ that will be denoted by  $\llangle \cdot , \cdot \rrangle^S$. 
For $v_1, v_2 \in T_sS$ it is defined by
\begin{equation}
\label{eq:metric-on-S}
\llangle v_1, v_2 \rrangle^S_s :=\llangle  \mbox{hor}_q(v_1), \mbox{hor}_q(v_2) \rrangle_q, \qquad q\in \pi^{-1}(s),
\end{equation}
and is locally given by $\sum_{i,j=1}^r K_{ij} ds^i \otimes ds^j$ with $K_{kl}$ 
defined by~\eqref{eq:local-S-metric}.
Similarly, the invariance of the potential energy $U$ induces a reduced potential
$U_S\in C^\infty(S)$ such that  $U=U_S\circ \pi$. Therefore, there is a well defined \defn{reduced Lagrangian} 
$\mathcal{L}:TS\to \R$,
of mechanical type, defined by
\begin{equation}
\label{eq:reduced-Lagrangian}
\mathcal{L}(s,\dot s) =\frac{1}{2}\llangle \dot s, \dot s \rrangle^S_s - U_S(s).
\end{equation}
Locally we have $\mathcal{L}(\dot s, s) = \frac{1}{2}\sum_{k,l=1}^r K_{kl}\dot s^k\dot s^l - U_S(s)$. 
As was announced above, the  reduced equations of motion take the form of a  mechanical system on $S$
which is subject to {\em gyroscopic forces}. These may be written in terms of the gyroscopic coefficients $C_{ij}^k$
as  (see e.g.~\cite{LGN18}):
\begin{equation}
\label{eq:reduced-Lag-form}
\frac{d}{dt} \left ( \frac{\partial \mathcal{L}}{\partial \dot s^i} \right ) - \frac{\partial \mathcal{L}}{\partial s^i}
=-\sum_{j,k=1}^rC_{ij}^k \dot s^j  \frac{\partial \mathcal{L}}{\partial \dot s^i}, \qquad i=1,\dots, r.
\end{equation}
Now  use the standard \defn{Legendre transformation}:
\begin{equation}
\label{eq:LegTransf}
p_i=\frac{\partial \mathcal{L}}{\partial \dot s^i} ,  \qquad i=1,\dots, r,
\end{equation}
and define the \defn{reduced Hamiltonian} $H:T^*S\to \R$
\begin{equation}
\label{eq:HamDef}
 H(s,p):=\sum_{j=1}^rp_j  \frac{\partial \mathcal{L}}{\partial \dot s^j}  - \mathcal{L} = \sum_{k,l=1}^r K^{kl}p_kp_l +U_S(s),
\end{equation}
where $K^{kl}$ are the entries of the inverse matrix of $K_{kl}$. It is a standard
exercise to show that  Eqs.~\eqref{eq:reduced-Lag-form} are equivalent to the 
 following first order system on $T^*S$:  
\begin{equation}
\label{eq:motionT*S}
\dot s^i = \frac{\partial H}{\partial p_i}, \qquad \dot
 p_i = -\frac{\partial H}{\partial s^i} -\sum_{j,k=1}^rC_{ij}^k p_k \frac{\partial H}{\partial 
 p_j}, \qquad i=1,\dots, r.
\end{equation}
Eqs.~\eqref{eq:motionT*S} give the local expression of the reduced vector field $\bar X_{nh}$ on
$T^*S$. As mentioned above, this vector field satisfies  ${\bf i}_{\bar X_{nh}} \Omega_{nh} =dH$ 
where the almost-symplectic 2-form
\begin{equation*}
\Omega_{nh}=\sum_{j=1}^r ds^j \wedge dp_j + \sum_{i<j} \sum_{k=1}^r C_{ij}^kp_k \, ds^i\wedge ds^j.
\end{equation*}
We refer the reader to~\cite{GNMarr2018} for an intrinsic construction of $\Omega_{nh}$ using the gyroscopic tensor. A different,
yet equivalent, approach to the construction of $\Omega_{nh}$ is taken in~\cite{EhlersKoiller}.

\subsection{$\phi$-simple Chaplygin systems, measure preservation and Hamiltonisation}

In this subsection we recall the recent results on Chaplygin Hamiltonisation from Garc\'ia-Naranjo~\cite{LGN18} 
and Garc\'ia-Naranjo and Marrero~\cite{GNMarr2018}\footnote{As shown in~\cite{GNMarr2018}, these results are equivalent
to certain sufficient conditions for Hamiltonisation given first by Stanchenko~\cite{Stanchenko} and Cantrijn et al.~\cite{CaCoLeMa}.}.
We begin with the following:

\begin{definition}[\cite{GNMarr2018}]
\label{D:sigma-gyro}
 A non-holonomic Chaplygin system is said to be \defn{$\phi$-simple}  if  there exists a function $\phi \in C^\infty(S)$ such that the gyroscopic 
 tensor ${\mathcal T}$ satisfies
 \begin{equation}\label{Ham-condition}
{\mathcal T}(Y, Z) =Z[\phi]Y - Y[\phi]Z, 
\end{equation}
for all $ Y, Z \in {\frak X}(S)$.
\end{definition}

The class of $\phi$-simple Chaplygin systems is quite special. As it turns out,  their reduced equations on $T^*S$
 always possess an invariant measure and 
allow a Hamiltonisation by Chaplygin's reducing multiplier method. 
More precisely:

\begin{theorem}[\cite{GNMarr2018}] 
\label{T:phi-simple}
\begin{enumerate}
\item The reduced equations of  motion~\eqref{eq:motionT*S} of a $\phi$-simple Chaplygin system possess
 the invariant measure $\mu = \exp (\sigma) \, \nu$, where $\nu$ is the Liouville measure
on $T^*S$ and $\sigma=(r-1)\phi$. 
\item The reduced equations of motion~\eqref{eq:motionT*S} of a $\phi$-simple Chaplygin system become Hamiltonian after the time reparametrisation  $dt =\exp(-\phi (s))\, d\tau$.
\end{enumerate}
\end{theorem}

The implication of item (i) of the theorem is clear. If the Chaplygin system under consideration is $\phi$-simple, then 
Eqs.~\eqref{eq:motionT*S} preserve the volume form on $T^*S$ whose local expression is
\begin{equation*}
\mu= \exp \left ( (1-r)\phi \right ) \, ds^1 \wedge \cdots \wedge ds^r \wedge dp_1 \wedge \cdots \wedge dp_r.
\end{equation*}

One interpretation of item (ii) 
 - followed in~\cite{GNMarr2018} -
 is that $\phi$-simplicity implies that the 2-form $\Omega_{nh}$ is conformally symplectic  with conformal factor $\exp (\phi(s))$. In other words, 
 the 2-form $\overline{ \Omega} := \exp (\phi(s))\Omega_{nh}$ is closed and hence symplectic. Hence, the rescaled
 vector field $Z:= \exp (-\phi(s))\bar X_{nh}$ satisfies 
 ${\bf i}_{Z} \overline{ \Omega}=dH$ and is therefore Hamiltonian (with respect to the symplectic
 structure $\overline{ \Omega}$). 
 
 Another, equivalent, interpretation of item (ii) of Theorem~\ref{T:phi-simple} - followed in~\cite{LGN18} - is obtained by defining the momentum rescaling:
 \begin{equation*}
 \tilde p_i= \exp(\phi (s))\, p_i, \quad i=1,\dots, r,
\end{equation*}
 and writing the reduced Hamiltonian in the new variables
 \begin{equation*}
\tilde H(s,\tilde p)= H(s,\exp(-\phi(s)) \tilde p).
\end{equation*}
Item (ii) of Theorem~\ref{T:phi-simple} states that, for a $\phi$-simple Chaplygin system, 
 Eqs.~\eqref{eq:motionT*S} are written in 
the $(s, \tilde p)$  variables in \defn{conformally Hamiltonian} form 
 \begin{equation}
 \label{eq:conf-Hamiltonian}
\frac{d s^i}{dt} 
= \exp(\phi (s))\ \frac{\partial \tilde H}{\partial \tilde  p_i}, \qquad \frac{d \tilde p_i}{dt}   = -\exp(\phi (s))\frac{\partial \tilde H}{\partial s^i}, \qquad i=1,\dots, r.
\end{equation}
The conformal factor $\exp(\phi (s))$ may be absorbed in the time reparametrisation $dt =\exp(-\phi (s))\, d\tau$ leading to the Hamiltonian 
system:
 \begin{equation*}
\frac{d s^i}{d\tau} 
= \frac{\partial \tilde H}{\partial \tilde  p_i}, \qquad  \frac{d \tilde p_i}{d\tau} = -\frac{\partial \tilde H}{\partial s^i}, \qquad i=1,\dots, r.
\end{equation*}

\section{Noether's Theorem for $\phi$-simple Chaplygin systems}
\label{S:Noether}

We now show  how additional - sometimes called internal - symmetries of $\phi$-simple Chaplygin systems lead to first integrals. This is a 
consequence of the  conformally Hamiltonian structure of their reduced equations. 
As we show below, the conserved quantities are simply a  rescaling by the conformal factor 
of the standard momentum map for Hamiltonian systems.

We begin by recalling  some standard notation. 
Suppose that the Lie group $A$, with Lie algebra  $\frak{a}$, acts 
 on $S$. 
 For $\xi\in \frak{a}$ we denote by $\xi_S \in \frak{X}(S)$ the \defn{infinitesimal 
generator} of $\xi$. Namely, $\xi_S$ is the vector field on $S$ defined by
\begin{equation*}
\xi_S(s) := \left . \frac{d}{dt} \right |_{t=0}  \exp(t\xi)\cdot s \in T_sS.
\end{equation*}

\begin{theorem}
\label{Th:Noether}
Consider a $\phi$-simple Chaplygin system and suppose that the  Lie group $A$ acts  on the 
shape space $S$ and leaves $\phi$ invariant.
\begin{enumerate}
\item If the reduced Lagrangian $\mathcal{L}:TS\to \R$ defined by~\eqref{eq:reduced-Lagrangian}
 is invariant under the tangent lifted action of $A$ to $TS$,
  then  the \defn{rescaled tangent bundle momentum map}
\begin{equation}
\label{eq:tangentbundlemommap}
  \mathcal{J}:TS\to  \frak{a}^*, \quad \mbox{defined by} \quad  \mathcal{J}(s,\dot s) ( \xi)  =  \exp(\phi(s)) \llangle \dot s  , \xi_S(s)  \rrangle^S_s,  \qquad \xi\in  \frak{a},
\end{equation}
is constant along the flow of the reduced equations~\eqref{eq:reduced-Lag-form}. 
\item If the reduced Hamiltonian $H:T^*S\to \R$ defined by~\eqref{eq:HamDef}
 is invariant under the cotangent lifted action of $A$ to $T^*S$,
  then  the \defn{rescaled cotangent bundle momentum map}
\begin{equation*}
  \mathcal{J}:T^*S\to  \frak{a}^*, \quad \mbox{defined by} \quad  \mathcal{J}(s,p)  =  
  \exp(\phi(s)) \langle p ,\xi_S(s) \rangle,  \qquad \xi\in  \frak{a},
\end{equation*}
where $\langle \cdot , \cdot \rangle$ denotes the duality pairing between $T^*S$ and $TS$, 
is constant along the flow of the reduced equations~\eqref{eq:motionT*S}. 

\item Items (i) and (ii) are equivalent via the Legendre transformation~\eqref{eq:LegTransf}.
\end{enumerate}
\end{theorem}

\begin{proof}
Let $\xi\in \frak{a}$ and suppose that in local coordinates
  $ \xi_S = \sum_{i=1}^r \xi^j(s) \frac{\partial}{\partial s^j}$.
 The rescaled tangent bundle momentum map is locally given by:
\begin{equation*}
 \mathcal{J}(s^i,\dot s^i) ( \xi)  =  \exp(\phi(s))  \sum_{i=1}^r \xi^j(s) \frac{\partial \mathcal{L}}{\partial \dot s^j},
\end{equation*}
which,  in view of  the Legendre transformation~\eqref{eq:LegTransf}, 
coincides with
\begin{equation}
\label{eq:RescaledMomentum}
 \mathcal{J}(s^i, p_i) ( \xi)  =  \exp(\phi(s))  \sum_{i=1}^r \xi^j(s) p_j,
\end{equation}
which is the local expression for the rescaled cotangent bundle momentum map.
So the definitions of $ \mathcal{J}$ in items (ii) and (iii) 
are  indeed matched by the Legendre transformation. The equivalence
between the invariance assumptions on $L$ and $H$ - with respect to the appropriate lifted action
of $A$ - is quite standard (see~e.g.~\cite{MaRa}).  We complete the proof by showing 
that, under the cotangent lift 
invariance assumption on $H$, $\mathcal{J}$ given by~\eqref{eq:RescaledMomentum} is indeed a first integral of Eqs.~\eqref{eq:motionT*S}. We begin by using the 
assumption of $\phi$-simplicity  to rewrite Eqs.~\eqref{eq:motionT*S} as
\begin{equation*}
 \label{eq:conf-Hamiltonianbis}
\frac{d s^i}{dt}
= \frac{\partial H}{\partial   p_i}, \qquad \frac{d  }{dt} \left (\exp(\phi) p_i  \right )= -\exp(\phi) \left ( 
\frac{\partial  H}{\partial s^i} +
\frac{\partial  \phi}{\partial s^i}\sum_{j=1}^r\frac{\partial  H}{\partial p_j}p_j \right ), \qquad i=1,\dots, r.
\end{equation*}
Indeed, a calculation based on the chain rule shows that the above system is equivalent to 
Eqs.~\eqref{eq:conf-Hamiltonian}. Therefore, using the above equations, we compute
\begin{equation}
\label{eq:NoetherThaux0}
\frac{d }{dt} \left ( \exp(\phi ) \sum_{i=1}^r\xi^i p_i  \right ) = - \exp(\phi )  \left [ 
\sum_{i=1}^r\left ( \frac{\partial  H}{\partial s^i}\xi^i - p_i\sum_{j=1}^r
\frac{\partial \xi^i}{\partial s^j} \frac{\partial H}{\partial   p_j}  \right ) +
  \left ( \sum_{j=1}^r\frac{\partial  H}{\partial p_j}p_j   \right )
 \left ( \sum_{i=1}^r \frac{\partial  \phi}{\partial s^i}\xi^i \right ) \right ].
\end{equation}

On the other hand, the $A$-invariance of $\phi$ implies
\begin{equation}
\label{eq:NoetherThaux1}
\xi_S[\phi]= \sum_{i=1}^r \xi^i \frac{\partial \phi }{\partial s^i} =0.
\end{equation}
Moreover, recall (see~e.g.~\cite{MaRa}) 
that the cotangent lift of $ \xi_S = \sum_{i=1}^r \xi^j(s) \frac{\partial}{\partial s^j}$ is
the vector field $ \xi_S^{T^*S}$ on $T^*S$ expressed in bundle coordinates as 
$ \xi_S^{T^*S} = \sum_{i=1}^r \xi^i \frac{\partial }{\partial s^i} -\sum_{i,j=1}^r
\frac{\partial \xi^i}{\partial s^j} p_i \frac{\partial }{\partial   p_j}$. Therefore, the assumption
that $H$ is invariant under the cotangent lift of $A$ to $T^*S$ implies
\begin{equation}
\label{eq:NoetherThaux2}
  \xi_S^{T^*S} [H]= \sum_{i=1}^r \xi^i \frac{\partial H }{\partial s^i} -\sum_{i,j=1}^r
\frac{\partial \xi^i}{\partial s^j} p_i \frac{\partial H}{\partial   p_j} =0.
\end{equation}
Substitution  of~\eqref{eq:NoetherThaux1} and~\eqref{eq:NoetherThaux2} into~\eqref{eq:NoetherThaux0}
proves the result.

\end{proof}

\section{The rubber Routh sphere} 
\label{S:Example}

Routh  considered the motion of a sphere, whose distribution of mass is axially symmetric, that rolls without slipping on the plane.
Here we enforce an additional {\em rubber} constraint that forbids spinning and consider the multi-dimensional generalisation of the
system. The 3D version of the problem was already considered by Borisov and coauthors in the works~\cite{BorMam2008, BorMamRubber} which treat
 more general problems of 3D bodies that roll without slipping or spinning over a surface.

\subsection{The 3D case}
\label{SS:3D}

Consider a sphere  that rolls without slipping or spinning on the plane. 
The orientation of the sphere is determined by an orthogonal matrix $R\in \SO(3)$ that relates a {\em body} fixed frame  
$\{E_1,E_2,E_3\}$  to an inertial or {\em space} 
frame $\{e_1,e_2,e_3\}$. We assume that space frame is chosen in such a way that the plane where the rolling takes place
is spanned by $e_1$ and $e_2$, and will denote 
the space coordinates of the geometric centre $O$ by
$\boldsymbol x=(x_1,x_2,b)^t$, where the constant
$b>0$ is the sphere's radius. The position and orientation of the sphere is hence completely determined by 
the pair $(R,(x_1,x_2))\in \SO(3)\times \R^2$ so the configuration space of the system is $Q=\SO(3)\times \R^2$.
To simplify the exposition, we will denote an element  $q\in Q$ as a pair $q=(R,x)\in \SO(3)\times \R^3$ with $x_3=b$. This amounts
to the identification of $Q$ with the embedded submanifold of  $\SO(3)\times \R^3$ defined by the holonomic
constraint  $x_3=b$.
%
%
%
%

Denote by $\boldsymbol{\omega}=(\omega_1,\omega_2,\omega_3)^t\in \R^3$ the angular velocity vector of the
sphere written in the {\em space} frame, and by  $\boldsymbol{\Omega}=R^{-1}\boldsymbol{\omega}=
(\Omega_1,\Omega_2,\Omega_3)^t\in \R^3$ the  
same  vector written  in the
 {\em body} frame. As is well known, these
  vectors correspond to the  right and left trivialisations of the tangent velocity vector
$\dot R\in T_R\SO(3)$ 
as follows
\begin{equation}
\label{eq:angVelocity3D}
\dot R R^{-1} = \begin{pmatrix} 0 & - \omega_3 & \omega_2 \\ 
 \omega_3 & 0 & -\omega_1 \\
  -\omega_2 & \omega_1 & 0 \end{pmatrix} \in \so (3), \qquad R^{-1}\dot R  = \begin{pmatrix} 0 & - \Omega_3 & \Omega_2 \\ 
 \Omega_3 & 0 & -\Omega_1 \\
  -\Omega_2 & \Omega_1 & 0 \end{pmatrix}  \in \so (3),
\end{equation}
where $\so (3)$, the space of skew-symmetric real $3\times 3$ matrices, is the Lie algebra of $\SO(3)$.

The no-slip rolling constraint  is written as 
\begin{equation}
\label{eq:3Drolling}
\dot x = b \boldsymbol{\omega} \times e_3, \qquad \mbox{or, equivalently,} \qquad 
\dot x  = b R(\Omega \times \gamma),
\end{equation}
where $\dot x=(\dot x_1, \dot x_2, 0)^t$,  $\times$ denotes the vector product in $\R^3$ and $\gamma:=R^{-1}e_3$ denotes the 
so-called {\em Poisson vector} that gives
 body coordinates of the vector $e_3$ that is normal
to the plane where the rolling takes place. 
The sphere is also subject to the no-spin or {\em rubber} constraint
\begin{equation}
\label{eq:3Drubber}
\omega_3 =0,  \qquad \mbox{or, equivalently,} \qquad  (\Omega,  \gamma)=0,
\end{equation}
where $(\cdot ,  \cdot)$ denotes the scalar euclidean product in $\R^3$. The first two components of 
 Eq.~\eqref{eq:3Drolling} together with 
Eq.~\eqref{eq:3Drubber}
define 3 independent nonholonomic constraints that determine a rank 2 distribution $D$ on $Q$.

Inspired by Routh~\cite{Routh}, we  assume that the mass distribution of the sphere is axially symmetric. The  body 
 frame $\{E_1,E_2,E_3\}$ is chosen with   origin 
at the centre of mass $C$ and with  $E_3$  aligned with  the 
axis of symmetry. This choice of body frame implies that  the inertia tensor  of the body is represented by a $3\times 3$ matrix of
the form  $\I=\mbox{diag}(I_1,I_1,I_3)$, with  principal moments of inertia $I_1, I_3>0$. We denote by $\ell$ the distance between $C$ and the 
geometric centre $O$ and assume that the  coordinates of $O$ in the body frame are $(0,0,-\ell)^t$, see Figure~\ref{F:Routh-sphere}. 
The space
coordinates of $C$ are hence given by the vector $\boldsymbol{u}=x+\ell R E_3$. 
Considering that $\| \dot {\boldsymbol{u}} \|^2=\| R^{-1}\dot {\boldsymbol{u}} \|^2$, where $\|\cdot \|$ is the Euclidean norm in 
$\R^3$, the Lagrangian of the 
system $L:TQ\to \R$, given by the kinetic minus the potential energy, may be written as
\begin{equation}
\label{eq:Lag3D}
L(R,\Omega, x, \dot x)=\frac{1}{2}(\I {\boldsymbol{\Omega}}, \boldsymbol{\Omega}) + \frac{m}{2}\|R^{-1} \dot x + \ell \Omega \times E_3
 \|^2  -m\mathcal{G}\ell \gamma_3 , 
\end{equation}
where  $m$ is the mass of the sphere, $ \mathcal{G}$ is the gravitational constant, and $\gamma_3$ denotes the third component of $\gamma$, i.e. $\gamma_3=(\gamma,E_3)$.

In Eq.~\eqref{eq:Lag3D}, and in what follows, we  write a generic element of $TQ$ as the quadruple 
$(R,\Omega, x, \dot x)\in \SO(3)\times \R^3 \times \R^3 \times \R^3$, which is possible by the identification of
$T\SO(3)$ with $\SO(3)\times \R^3 $ via the left trivialisation, and the  embedding $T\R^2\hookrightarrow \R^3\times \R^3$,
$((x_1,x_2),(\dot x_1,\dot x_2))\mapsto ((x_1,x_2,b),(\dot x_1,\dot x_2,0))$, induced from the holonomic constraint $x_3=b$.

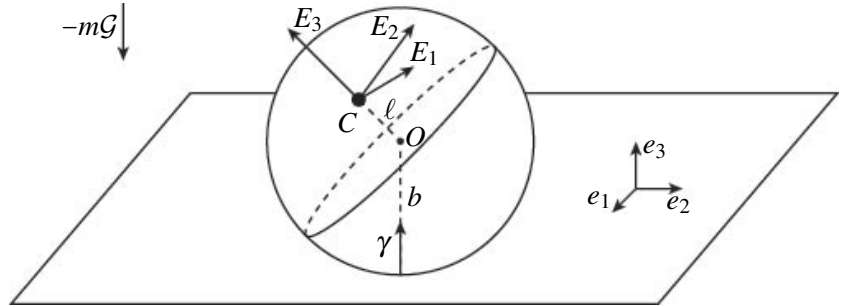
\begin{figure}[h]
\centering
\includegraphics[totalheight=4cm]{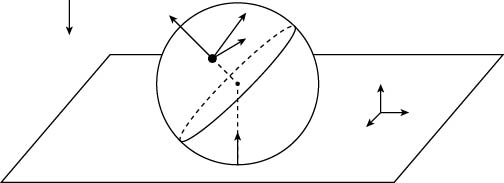}
 \put (-95,38) {$e_1$} \put (-65,36) {$e_2$} \put (-74,57) {$e_3$} \put (-164,60) {$O$}  \put (-189,65) {$C$}  \put (-163,37) {$b$} \put (-162,92) {$E_1$} \put (-177,102) {$E_2$}  \put (-172,70) {$\ell$} 
  \put (-206,105) {$E_3$}  \put (-175,20) {$\gamma$}   \put (-294,102) {$-m\mathcal{G}$}
\caption{Axisymmetric sphere on the plane. $E_3$ points along the axis of symmetry and $C$ is the centre of mass.}
\label{F:Routh-sphere}
\end{figure}

The evolution of the system is clearly independent of horizontal translations and rotations of the space frame about $e_3$.
 This corresponds to a symmetry action of the 
euclidean group $G=\SE(2)$  on the configuration space $Q=\SO(3)\times \R^2$ as we now show. We represent the group $G=\SE(2)$
as the Lie subgroup of $\mathrm{GL}(4,\R)$ consisting of matrices of the form
\begin{equation*}
g=\left ( \begin{array}{ccc|c} & & &  \\ & h & & y \\ & & &  \\ \hline 0 & 0 & 0 & 1\end{array} \right ), \qquad \mbox{where} \qquad 
y =\begin{pmatrix} y_1 \\  y_2 \\ 0   \end{pmatrix} \in \R^3, \qquad \mbox{and} \qquad h=\left ( \begin{array}{c|c} \tilde h &   \begin{array}{c} 0 \\ 0 \end{array}  \\    \hline  \begin{array}{cc} 0& 0 \end{array}
 & 1 \end{array} \right )\in \SO(3), 
\end{equation*}
with $ \tilde h\in \SO(2)$. The action of $g\in \SE(2)$ given  above on an element $(R,x)\in \SO(3)\times \R^3 $ is 
\begin{equation}
\label{eq:action}
g\cdot (R,x)= (hR,hx+y).
\end{equation}
This action restricts to $Q$ since it preserves the holonomic constraint $x_3=b$.

\begin{proposition}
\label{P:ChapSyst}
The problem of the rubber Routh sphere that rolls without slipping or spinning  on the plane 
is a Chaplgyin system with  $G=\SE(2)$ acting on $Q$ via Eq.~\eqref{eq:action}.
\end{proposition}

\begin{remark}
\label{rmk:ChapSyst}
Proposition~\ref{P:ChapSyst} is valid even if the sphere fails to be axially symmetric. In fact, it continues to hold for the problem
of an arbitrary rubber smooth convex body that rolls without slipping or spinning on the plane.
\end{remark}

\begin{proof}
The tangent lifted  action is
\begin{equation*}
g\cdot (R,\Omega,x,\dot x)= (hR,\Omega,hx+y,h\dot x). 
\end{equation*}
Moreover, given that $h^{-1}e_3=e_3$, it follows that the Poisson vector $\gamma$ is invariant. Using this, and the above expression
for the tangent lift, it is immediate to see that  both the rolling~\eqref{eq:3Drolling} and 
rubber~\eqref{eq:3Drubber} constraints are invariant. Similarly, one checks that the kinetic and potential energies of the 
Lagrangian~\eqref{eq:Lag3D} are invariant so the conditions (i) and (ii) in Definition~\ref{D:Chap} hold. In order to check that
condition (iii) in Definition~\ref{D:Chap} also holds, note that the Lie algebra $\se(2)$ in our representation is spanned by
the $4\times 4$ matrices
\begin{equation*}
 \xi_1=   \left ( \begin{array}{c|c}  \begin{array}{ccc}  &  &   \\  & 0 &  \\  &  &  \end{array}  &  \begin{array}{c} 1  \\  0 \\ 0  \end{array} \\  \hline 0 & 0 \end{array} \right ),
 \quad
 \xi_2=   \left ( \begin{array}{c|c}  \begin{array}{ccc}  &  &   \\  & 0 &  \\  &  &  \end{array}  &  \begin{array}{c} 0  \\  1 \\ 0  \end{array} \\  \hline 0 & 0 \end{array} \right ),
 \quad
  \xi_3= \left ( \begin{array}{c|c} \begin{array}{ccc} 0 & -1 & 0  \\ 1 & 0 & 0 \\ 0 & 0 & 0 \end{array} & 0 \\  \hline 0 & 0 \end{array} \right ).
\end{equation*}
The infinitesimal generators of $\xi_1$ and $\xi_2$ are  vector fields  on $Q$ that correspond to pure translations along the $x_1$
and $x_2$ axis respectively, which violate the rolling constraint~\eqref{eq:3Drolling}. On the other hand, 
the infinitesimal generator of $\xi_3$
is a vector field on $Q$ having constant $\omega_3=1$, which violates the rubber constraint~\eqref{eq:3Drolling}.
Hence, the group orbit is transversal to the constraint distribution and, by a dimension count, the  condition (iii) in 
Definition~\ref{D:Chap} is also verified. 
\end{proof}

The shape space $S=(\SO(3)\times \R^2)/\SE(2)$ is diffeomorphic to the two dimensional sphere $\Ss^2$ and the orbit projection
is
\begin{equation}
\label{eq:pi-proj3D}
\pi:\SO(3)\times \R^2 \to \Ss^2, \qquad (R,x)\mapsto \gamma,
\end{equation}
where we recall that $\gamma= R^{-1}e_3\in \R^3$ is the Poisson vector. Note that we realise 
\begin{equation*}
\Ss^2= \left \{ \gamma=(\gamma_1, \gamma_2,\gamma_3)^t \in \R^3 \, : \, \gamma_1^2+\gamma_2^2+\gamma_3^2 =1 \right  \}.
\end{equation*}
It follows from our discussion in section~\ref{S:ChapSyst} 
that the reduced equations of motion are defined on the cotangent bundle $T^*\Ss^2$. 

\begin{theorem}
\label{T:phiSimpleRouth3D}
The problem of the rubber Routh sphere that rolls without slipping or spinning  on the plane 
is $\phi$-simple with $\phi: \Ss^2\to \R$ given by
\begin{equation*}
\phi(\gamma)= -\frac{1}{2}\ln \left ( I_1\gamma_3^2 +I_3(1-\gamma_3^2)+ m(b+\ell \gamma_3)^2 \right ).
\end{equation*}
\end{theorem}
It follows from item (i) in Theorem~\ref{T:phi-simple}, that the reduced equations on $T^*\Ss^2$ possess the invariant measure: 
\begin{equation*}
\mu =\frac{1}{\sqrt{ I_1\gamma_3^2 +I_3(1-\gamma_3^2)+ m(b+\ell \gamma_3)^2 }} \, \nu,
\end{equation*}
where $\nu$ is the  Liouville volume form  on $T^*\Ss^2$.  Additionally,  item (ii) in Theorem~\ref{T:phi-simple} implies that 
  the reduced system on $T^*\Ss^2$ is conformally  Hamiltonian with time reparametrisation:
\begin{equation*}
dt=\sqrt{ I_1\gamma_3^2 +I_3(1-\gamma_3^2)+ m(b+\ell \gamma_3)^2 } \,\, d\tau.
\end{equation*}
\begin{remark}
The invariant measure for the problem  was
first given by Borisov and Mamaev in~\cite{BorMam2008} (see also~\cite{BorMamRubber}).
The Hamiltonisation of the system may be deduced as a consequence of the celebrated Chaplygin's 
Reducing Multiplier Theorem~\cite{ChapRedMult}
since the shape space $\Ss^2$ has dimension 2. For the multi-dimensional version
of the problem considered below,  these properties can no longer be 
deduced from known results and we will rely on  Theorem~\ref{T:phi-simple}.
\end{remark}

We do not present a proof of  Theorem~\ref{T:phiSimpleRouth3D} since it is  a particular instance of 
Theorem~\ref{T:phiSimpleRouthnD} below.

\subsection{The $n$D case}
\label{SS:nD}

We consider a multi-dimensional  generalisation of the problem considered in the previous section. Namely,
an $n$-dimensional rigid body of spherical shape, with axially symmetric distribution of mass,
 that rolls without slipping or spinning on a horizontal (with respect to 
gravity) hyperplane on 
$\R^n$.

The orientation of the sphere is determined by a rotation matrix $R\in \SO(n)$ that specifies the attitude of the sphere by 
relating a body fixed frame 
$\{E_1,\dots, E_n\}$ and a space frame $\{e_1, \dots, e_n\}$. In analogy with the 3D case, we assume that the
rolling takes place on the hyperplane spanned by $\{e_1, \dots, e_{n-1}\}$ and that the geometric centre $O$ of the sphere
has  space coordinates $x=(x_1, \dots, x_{n-1},b)$, where the constant $b>0$ 
is the sphere's radius. We will also assume, as in the 3D case, that the
body frame has its origin at the centre of mass $C$ and $E_n$ is aligned with the symmetry axis of the sphere.
The orientation of $E_n$ is such that the body coordinates of $O$ are $(0,\dots, 0,-\ell)$.
The configuration space of the problem is $Q=\SO(n)\times \R^{n-1}$. In analogy to 
 the 3D case, we will work with the embedding of $Q$ in
$\SO(n)\times \R^n$ defined by the holonomic constraint $x_n=b$.

As is well known, for $n>3$ the angular velocity can no longer be 
represented as a vector, but rather as an element in the Lie algebra $\so(n)$ of $\SO(n)$. We denote by $\omega\in \so(n)$ the
representation of the  angular velocity in the {\em space} frame and by $\Omega  \in \so(n)$
its representation in the  {\em body} frame. 
These are related to the right and left trivialisation of the tangent vector $\dot R\in T_R\SO(n)$ by
\begin{equation}
\label{eq:angVelocitynD}
\omega = \dot R R^{-1} \in \so (n), \qquad \Omega= R^{-1}\dot R   \in \so (n),
\end{equation}
and satisfy $\omega=\Ad_R\Omega$, where $\Ad_R:\so(n)\to \so(n)$ is the adjoint operator.

The constraint of rolling without slipping  is that the contact point of the sphere with the hyperplane $x_n=0$ has 
zero velocity at every time, and is expressed as the following natural generalisation of~\eqref{eq:3Drolling}:
\begin{equation}
\label{eq:nDrolling}
\dot x = b \boldsymbol{\omega} e_n, \qquad \mbox{or, equivalently,} \qquad 
\dot x  = b R \Omega \gamma,
\end{equation}
where $\dot x=(\dot x_1, \dots,  \dot x_{n-1}, 0)^t$, and the  Poisson vector $\gamma=(\gamma_1,\dots, \gamma_n)^t\in \R^n$
 is now given by  $\gamma:=R^{-1}e_n$.
On the other hand, the generalisation of the no-spin rubber constraint~\eqref{eq:3Drubber} is that the space representation of
the angular velocity satisfies
\begin{equation}
\label{eq:nDrubber}
\omega_{ij}=0, \qquad \mbox{for all} \qquad i,j=1, \dots, n-1.
\end{equation}
In other words, $\omega$ has the form
\begin{equation*}
\omega = \left ( \begin{array}{c|c}{\bf  0} & \begin{array}{c} \omega_{1n} \\ \vdots \\ \omega_{n-1}  \end{array} \\ \hline
 \begin{array}{ccc} -\omega_{1n} & \hdots &  - \omega_{n-1}  \end{array} & 0
\end{array} \right ),
\end{equation*}
where ${\bf 0}$ above denotes  the $(n-1)\times (n-1)$ zero matrix. The constraints~\eqref{eq:nDrubber} were considered
 by Jovanovi\'c~\cite{JovaRubber} in the
 treatment of the multi-dimensional rubber Chaplygin sphere. They 
 generalise  the 3D rubber constraint~\eqref{eq:3Drubber} in the following sense:
rotations of the sphere that occur on 2-dimensional planes that do not contain 
 the normal vector $e_n$  to the hyperplane where the rolling takes place are forbidden.

Our next step is to give a multi-dimensional  generalisation of  the Lagrangian~\eqref{eq:Lag3D}. 
For this matter we recall  that for an $n$-dimensional rigid body the inertia tensor $\I$ of the body is an operator
\begin{equation}
\label{eq:InTensorGen}
\I:\so(n) \to \so (n), \qquad \I(\Omega)=\J\Omega + \Omega\J,
\end{equation}
where $\J$ is the so-called mass tensor of the body, which is a symmetric and positive definite $n\times n$ matrix (see e.g. \cite{Ratiu80}). Our assumption that the mass distribution is axially symmetric, and that the $E_n$ axis of the body frame 
is aligned with the symmetry axis, imply that, with respect to our choice of body frame, the mass tensor has the form
\begin{equation}
\label{eq:MassTensorGen}
\J=\mbox{diag}(J_1, \dots, J_1, J_n), \qquad J_1, J_n>0.
\end{equation}

Similar to  our treatment of the 3D case, we shall represent elements of $TQ=T(\SO(n)\times \R^{n-1})$ as  quadruples 
$(R,\Omega, x ,\dot x)\in \SO(n)\times \so(n)\times \R^n\times \R^n$ with $x_n=b$ and $\dot x_n=0$. This is done by identifying $T\SO(n)=\SO(n)\times \so(n)$ via
the left trivialisation, and by embedding $T\R^{n-1}\hookrightarrow  \R^n\times \R^n$ putting $x_n=b$ 
and $\dot x_n=0$
The Lagrangian of the multi-dimensional system $L:TQ\to \R$ is
\begin{equation}
\label{eq:Lag-nD}
 L(R,\Omega, x ,\dot x)=\frac{1}{2} ( \I \Omega, \Omega  )_\kappa  
+ \frac{m}{2}\|R^{-1} \dot x + \ell \Omega  E_n
 \|^2  -m\mathcal{G}\ell \gamma_n, 
\end{equation}
where $\| \cdot \|$ is the euclidean norm in $\R^n$, $\dot x=(\dot x_1, \dots, \dot x_{n-1}, 0)^t$ 
and $( \cdot, \cdot)_\kappa$ is the Killing metric in $\so(n)$:
\begin{equation*}
(\xi, \eta)_\kappa =-\frac{1}{2}\tr(\xi \eta).
\end{equation*}
In~\eqref{eq:Lag-nD} we continue to denote by $\ell$ the distance of the centre of mass $C$ to the geometric centre $O$.

In analogy  to the 3D case, there is a symmetry action of the group $G=\SE(n-1)$ which we represent
 as the Lie subgroup of $\mathrm{GL}(n+1,\R)$ consisting of matrices of the form
\begin{equation*}
g=\left ( \begin{array}{ccc|c} & & &  \\ & h & & y \\ & & &  \\ \hline 0 & 0 & 0 & 1\end{array} \right ), \qquad \mbox{where} \qquad 
y =\begin{pmatrix} y_1 \\  \vdots \\ y_{n-1} \\ 0   \end{pmatrix} \in \R^n, \qquad \mbox{and} \qquad h=\left ( \begin{array}{c|c} \tilde h &   \begin{array}{c} 0 \\ 0 \end{array}  \\    \hline  \begin{array}{cc} 0& 0 \end{array}
 & 1 \end{array} \right )\in \SO(n), 
\end{equation*}
with $ \tilde h\in \SO(n-1)$. The action of $g\in \SE(n-1)$ given  above on an element $(R,x)\in \SO(n)\times \R^n$ looks 
identical to Eq.~\eqref{eq:action}, namely
\begin{equation}
\label{eq:actionnD}
g\cdot (R,x)= (hR,hx+y).
\end{equation}
As in the 3D case, the action restricts to $Q$ since the holonomic constraint $x_n=b$ is invariant. In analogy with 
Proposition~\ref{P:ChapSyst}
we have:
\begin{proposition}
\label{P:ChapSystnD}
The $n$-dimensional generalisation of the  problem of the rubber Routh sphere that rolls without slipping or spinning  on a hyperplane 
is a Chaplgyin system with  $G=\SE(n-1)$ acting on $Q$ via Eq.~\eqref{eq:actionnD}.
\end{proposition}
The proof is analogous to that of Proposition~\ref{P:ChapSyst} and we omit the details. Also,
 in analogy with Remark~\ref{rmk:ChapSyst}, we mention that the  conclusion of Proposition~\ref{P:ChapSystnD}
  is independent of our symmetry assumptions on the
 mass distribution of the sphere and also applies to general rubber multi-dimensional  convex rigid bodies 
 that roll without slipping or spinning
 on a horizontal hyperplane in $\R^n$.

The  shape space of the system
 $S=(\SO(n)\times \R^{n-1})/\SE(n-1)$ is diffeomorphic to the $n-1$ dimensional sphere $\Ss^{n-1}$,
  and the orbit projection~\eqref{eq:pi-proj3D}, valid in 3D, 
generalises automatically to
\begin{equation*}
\pi:\SO(n)\times \R^{n-1} \to \Ss^{n-1}, \qquad (R,x)\mapsto \gamma,
\end{equation*}
where we recall that in the  the Poisson vector $\gamma= R^{-1}e_n\in \R^{n}$, and we realise $\Ss^{n-1}$ by its embedding
in $\R^n$:
\begin{equation}
\label{eq:Sn-1}
\Ss^{n-1}= \left \{ \gamma=(\gamma_1, \dots ,\gamma_n)^t \in \R^n \, : \, \gamma_1^2+\dots+\gamma_n^2 =1 \right  \}. 
\end{equation}
As a consequence of our discussion in section~\ref{S:ChapSyst}, 
 the $G=\SE(n-1)$-reduced equations of motion live on the cotangent bundle $T^*\Ss^{n-1}$. We now state our main 
result which, in view of Theorem~\ref{T:phi-simple},
 implies that the multi-dimensional rubber Routh sphere has an invariant measure and allows a Hamiltonisation.

\begin{theorem}
\label{T:phiSimpleRouthnD}
The $n$-dimensional  rubber Routh sphere that rolls without slipping or spinning  on a horizontal hyperplane  
is $\phi$-simple with $\phi: \Ss^{n-1}\to \R$ given by
\begin{equation}
\label{eq:phi-nD}
\phi(\gamma)= -\frac{1}{2}\ln \left ( 2 J_1 +(J_n-J_1)\gamma_n^2+ m(b+\ell \gamma_n)^2 \right ).
\end{equation}
\end{theorem}
\begin{remark}
The conclusion of the above theorem is consistent with Theorem~\ref{T:phiSimpleRouth3D} by noting that, in the 3D case,
 the principal moments
of inertia $I_1, I_3$ are related to the entries $J_1, J_3$ of the mass tensor $\J$ by the relations $I_1=J_1+J_3$ and $I_3=2J_1$.
\end{remark}
As a direct consequence of Theorem~\ref{T:phi-simple} and 
Theorem~\ref{T:phiSimpleRouthnD} we obtain:
\begin{corollary}
\label{Cor:MeasHam}
The $\SE(n-1)$-reduced equations on $T^*\Ss^{n-1}$ of the 
$n$-dimensional  rubber Routh sphere that rolls without slipping or spinning  on a horizontal hyperplane 
possess the invariant measure: 
\begin{equation*}
\mu =\ \left ( 2 J_1 +(J_n-J_1)\gamma_n^2+ m(b+\ell \gamma_n)^2 \right )^{\frac{2-n}{2}} \, \nu,
\end{equation*}
where $\nu$ is the  Liouville volume form  on $T^*\Ss^{n-1}$. Moreover, such reduced system is conformally  Hamiltonian with time reparametrisation:
\begin{equation*}
dt=\sqrt{  2 J_1 +(J_n-J_1)\gamma_n^2+ m(b+\ell \gamma_n)^2 } \,\, d\tau.
\end{equation*}
\end{corollary}
\begin{proof}
 The conclusion about the invariant measure follows from~Eq.~\eqref{eq:phi-nD}
and  item  (i) of Theorem~\ref{T:phi-simple} (putting $r=n-1$).
The conformally Hamiltonian structure of the equations of motion follows from item (ii) of Theorem~\ref{T:phi-simple}.
\end{proof}

The rest of the paper is devoted to prove Theorem~\ref{T:phiSimpleRouthnD} via a coordinate calculation. The strategy is to  follow
 the steps outlined in section~\ref{SS:local-expressions} to compute the gyroscopic coefficients $C_{ij}^k$. We will work with the
local coordinates\footnote{throughout this section, and in contrast with our notation in section~\ref{S:ChapSyst}, 
we use sub-indices instead of super-indices on the
coordinates on $S$.} $(s_1, \dots, s_{n-1})$ valid on the northern  $\Ss_+^{n-1}$, or 
southern   $\Ss_-^{n-1}$, hemispheres of $\Ss^{n-1}$
by the relations:
\begin{equation*}
\gamma_1=s_1,\quad \dots \quad ,  \quad \gamma_{n-1}=s_{n-1}, \quad  \gamma_n=\pm\sqrt{1-s_1^2-\dots -s_{n-1}^2}. 
\end{equation*}
Associated  to the embedding of $\Ss^{n-1}$ in $\R^n$, there is an  embedding of $T\Ss^{n-1}$ in $\R^n\times \R^n$
given by
\begin{equation}
\label{eq:TSn-1}
T\Ss^{n-1} = \left \{ (\gamma, v)\in\R^n\times \R^n \, : \, \|\gamma\|=1, \quad  (\gamma, v)_{\R^n}=0 \right \}, 
\end{equation}
where $ (\cdot , \cdot )_{\R^n}$ is the Euclidean scalar product in $\R^n$.
Under the above identification,  and regardless of the hemisphere under consideration, 
the coordinate vector fields $\frac{\partial}{\partial s_i}$ are given in terms of the  canonical vectors $E_1, \dots, E_n$, 
by
\begin{equation}
\label{eq:ds_i-example}
\frac{\partial}{\partial s_i}=E_i -\frac{\gamma_i}{\gamma_n} E_n, \qquad i=1, \dots, n-1.
\end{equation}
 
 For the rest of the section, given $(R,x)\in \SO(n)\times \R^n$, we identify 
 \begin{equation*}
T_{(R,x)}(\SO(n)\times \R^n)=\so(n)\times \R^n,
\end{equation*}
using the left trivialisation of $T_R\SO(n)$ and the usual identification $T_x\R^n=\R^n$. 
Therefore, a vector field on $Q$ is represented as an assignment that to a 
pair $(R,x)\in \SO(n)\times \R^n$ with $x_n=b$, associates a pair $(\xi(R,x),V(R,x))\in \so(n)\times \R^n$ that satisfies $V_n(R,x)=0$.
We will also find it useful to denote 
\begin{equation*}
y\wedge z := yz^t -zy^t\in \so(n), \qquad \mbox{for column vectors} \qquad  y,z\in \R^n.
\end{equation*}

 In accordance with the notation of section~\ref{SS:local-expressions} we denote by $h_i$ the horizontal
 lift of the coordinate vector field $\frac{\partial}{\partial s_i}$. Namely,
 \begin{equation*}
  h_i(R,x):= \mbox{hor}_{(R,x)} \left ( \frac{\partial}{\partial s_i}\right ).
\end{equation*}
The following proposition gives an explicit expression for $h_i(R,x)$.
\begin{proposition}
Let $\gamma\in \Ss_\pm^{n-1}$ and $(R,x)\in \pi^{-1}(\gamma)$, (i.e. $\gamma=R^{-1}e_n$). The horizontal lift
\begin{equation}
\label{eq:defXi}
h_i(R,x) =\left (  \gamma\wedge \left ( E_i  -\frac{\gamma_i}{\gamma_n}E_n \right ), aR \left (- E_i  + \frac{\gamma_i}{\gamma_n}E_n 
\right )  \right ) \in \so(n) \times \R^n,
\qquad i=1,\dots, n-1.
\end{equation}
%
\end{proposition}
\begin{proof}
The rubber  constraints~\eqref{eq:nDrubber} imply $\omega=e_n\wedge y$ for a vector $y\in \R^n$ that may be assumed to be
perpendicular to $e_n$. 
Hence, $$\Omega=\mbox{Ad}_{R^{-1}}(e_n\wedge a)=(R^{-1}e_n)\wedge ( R^{-1}y)=\gamma\wedge v,$$
where $v=R^{-1}y$ is perpendicular to $\gamma$. On the other hand, differentiating $\gamma=R^{-1}e_n$ gives $\dot \gamma =-\Omega \gamma$.
Whence,  $\dot \gamma = -(\gamma \wedge v)\gamma=v$ and we conclude that 
\begin{equation}
\label{eq:Omega-as-gammawedgedotgamma}
\Omega=\gamma \wedge \dot \gamma.
\end{equation}
The rolling constraint~\eqref{eq:nDrolling} then implies
\begin{equation*}
\dot x  = b R (\gamma \wedge \dot \gamma) \gamma =-bR\dot \gamma.
\end{equation*}
Therefore, we get the following expression for the horizontal lift
\begin{equation}
\label{eq:horlift}
\mbox{hor}_{(R,x)} ( \dot \gamma ) = \left ( \gamma\wedge \dot \gamma, -bR\dot \gamma \right )\in \so(n)\times \R^n, 
\quad \mbox{where} \quad
\dot \gamma \in T_\gamma \Ss^{n-1} \quad \mbox{and} \quad (R,x)\in \pi^{-1}(\gamma).
\end{equation}
The result then follows by using~\eqref{eq:ds_i-example}.
\end{proof}

%
%
%
%
The following lemma gives expressions, that involve the horizontal lifts $h_i$ 
and the kinetic energy metric $\llangle \cdot , \cdot \rrangle$, that will be used
below to compute the gyroscopic coefficients $C_{ij}^k$. Its proof is 
postponed to  Appendix~\ref{app:proofs-lemmas}.
%
%
 \begin{lemma}
\label{l:MetricXi}
For $i,j,k,l\in \{1, \dots, n-1\}$ we have
\begin{equation}
\begin{split}
\label{eq:Kij}
K_{kl} & = \llangle h_k, h_l\rrangle  \\
&= \left ( 2 J_1 +(J_n-J_1)\gamma_n^2+ m(b+\ell \gamma_n)^2 \right ) \, \delta_{kl} +
 \left (  \frac{J_1+J_n}{\gamma_n^2} + J_n-J_1 + m\left ( \left ( \frac{b}{\gamma_n} +\ell \right )^2+\frac{\ell^2}{ \gamma^2_n} \right )
 \right )\gamma_k\gamma_l,
 \end{split}
\end{equation}
and
\begin{equation}
\label{eq:XiXjXl}
\llangle [h_i, h_j ] , h_l \rrangle  = \left ( J_n-J_1 +m\ell \left ( \ell + \frac{b}{\gamma_n} \right ) \right )
 \left ( \gamma_j \delta_{il} -\gamma_i\delta_{jl} \right ),
\end{equation}
where $\delta_{ij}$ is the Kronecker delta.
\end{lemma}

We are now ready to prove the following lemma that gives 
explicit expressions for the gyroscopic coefficients $C_{ij}^k$ in our coordinates.
\begin{lemma}
\label{L:Gyro-coeff-Veselova}
For $i,j,k\in \{1, \dots, n-1\}$ we have
\begin{equation}
\label{eq:Cijk}
C_{ij}^k  = \frac{ \left ( J_n-J_1 +m\ell \left ( \ell + \frac{b}{\gamma_n} \right ) \right )
 \left ( \gamma_j \delta_{ik} -\gamma_i\delta_{jk} \right )}{2 J_1 +(J_n-J_1)\gamma_n^2+ m(b+\ell \gamma_n)^2}.
\end{equation}
\end{lemma}

\begin{proof} 
Using 
\begin{equation*}
\sum_{k=1}^{n-1} \left ( \gamma_j \delta_{ik} -\gamma_i\delta_{jk} \right ) \delta_{kl}= \gamma_j \delta_{il} -\gamma_i\delta_{jl} \qquad
\mbox{and}
\qquad \sum_{k=1}^{n-1} \left ( \gamma_j \delta_{ik} -\gamma_i\delta_{jk} \right ) \gamma_k\gamma_l =0,
\end{equation*}
it follows, in view of \eqref{eq:Kij} and \eqref{eq:XiXjXl},   that $C_{ij}^k$ as given by~\eqref{eq:Cijk} satisfy
\begin{equation*}
\sum_{k=1}^{n-1}K_{kl}C_{ij}^k=\llangle [h_i, h_j ] , h_l \rrangle, \qquad i,j,k,l\in \{1, \dots, n-1\}.
\end{equation*}
In other words, the expressions~\eqref{eq:Cijk} for $C_{ij}^k$ are the unique solution
 to the system~\eqref{eq:gyr-coeff} that determines the gyroscopic coefficients.
\end{proof}

We are now ready to present:
\begin{proof}[Proof of Theorem~\ref{T:phiSimpleRouthnD}]
Lemma~\ref{L:Gyro-coeff-Veselova} implies
\begin{equation}
\label{eq:auxProofVesGyroT}
\mathcal{T} \left (\frac{\partial}{\partial s_i} ,\frac{\partial}{\partial s_j} \right ) = 
 \frac{ \left ( J_n-J_1 +m\ell \left ( \ell + \frac{b}{\gamma_n} \right ) \right )}{2 J_1 +(J_n-J_1)\gamma_n^2+ m(b+\ell \gamma_n)^2}
\left (  s_j \frac{\partial}{\partial s_i} -s_i \frac{\partial}{\partial s_j} \right ), \qquad 1\leq i,j\leq n-1.
\end{equation}
Considering that $\gamma_n=\pm \sqrt{1-s_1^2-\dots -s_{n-1}^2}$, we have $\partial \gamma_n / \partial s_k =- s_k/\gamma_n$
 for  $1\leq k \leq n-1$, and hence, from the expression~\eqref{eq:phi-nD} for 
 $\phi:\Ss^{n-1}\to \R$  we compute
 \begin{equation*}
\begin{split}
\frac{\partial \phi}{\partial s_k}   = 
 \frac{1}{2}\left (  \frac{2(J_n-J_1)\gamma_n 
 +2m\ell(b+\ell \gamma_n) }{2 J_1 +(J_n-J_1)\gamma_n^2+
  m(b+\ell \gamma_n)^2} \right ) \frac{s_k}{\gamma_n} 
  = 
  \left ( \frac{ J_n-J_1  + m\ell(\ell + \frac{b}{\gamma_n}) }{2 J_1 +(J_n-J_1)\gamma_n^2+ m(b+\ell \gamma_n)^2} \right )s_k,
 \qquad  1\leq k \leq n-1.
 \end{split}
\end{equation*}
%

 Therefore, Eq.~\eqref{eq:auxProofVesGyroT} may be rewritten as
 \begin{equation*}
\label{eq:gyro-T-Veselova}
\mathcal{T} \left (\frac{\partial}{\partial s_i} ,\frac{\partial}{\partial s_j} \right ) = \frac{\partial \phi}{\partial s_j}\frac{\partial}{\partial s_i} -
\frac{\partial \phi}{\partial s_i}\frac{\partial}{\partial s_j}, \qquad 1\leq i,j\leq n-1.
\end{equation*}
The  above expression, together with the 
 tensorial properties of $\mathcal{T}$, shows that the $\phi$-simplicity condition   \eqref{Ham-condition} 
 holds on the open dense subset of $\Ss^{n-1}$ where $\gamma_n\neq 0$. By continuity, it holds on all of $\Ss^{n-1}$.
 \end{proof}

\begin{remark}
We note that the notion of $\phi$-simplicity, and hence also our conclusions about  measure preservation
and Hamiltonisation, only depends on the kinetic energy and the constraints and does not involve
 the gravitational potential. This is a consequence of the
   {\em weak Noetherianity} of these concepts (see~\cite{GNMarr2018}). 
\end{remark}

\subsection{First integrals }
\label{SS:Integrability}

In this section we use Theorem~\ref{Th:Noether} to prove that 
\begin{equation}
\label{eq:FirstIntegralsRouth}
F_{ij}:=\sqrt{  2 J_1 +(J_n-J_1)\gamma_n^2+ m(b+\ell \gamma_n)^2 } \, \Omega_{ij}, \qquad 1\leq i,j\leq n-1,
\end{equation}
are first integrals of the system. In 3D, there is only one such integral whose existence had been proven by Borisov 
and Mamaev~\cite{BorMam2008} and, considering that $T^*\Ss^{2}$ has 2 degrees of freedom, it is sufficient to conclude
integrability of the problem. The question of integrability in $n$D will be addressed in a forthcoming publication.

We begin by noting that, in view of expression~\eqref{eq:horlift} for the horizontal lift of $\dot \gamma \in T_\gamma \Ss^{n-1}$
and the expression for the Lagrangian~\eqref{eq:Lag-nD}, the reduced Lagrangian $\mathcal{L}:T\Ss^{n-1} \to \R$ is given by
\begin{equation*}
\mathcal{L}(\gamma,  \dot \gamma) = \frac{1}{2} \left ( \I (\dot \gamma \wedge \gamma) , \dot \gamma \wedge \gamma \right )_\kappa 
+\frac{m}{2} \left \| (b+\ell \gamma_n) \dot \gamma - \ell \dot \gamma_n \gamma \right \|^2 -  m\ell \mathcal{G}\gamma_n,
\end{equation*}
which, using the specific form of the inertia tensor $\I$ given by Eqs.~\eqref{eq:InTensorGen} and \eqref{eq:MassTensorGen},
simplifies to
\begin{equation}
\label{eq:red-Lag-Routh}
\mathcal{L}(\gamma,  \dot \gamma) = \frac{1}{2} \left ( 2J_1 +(J_n-J_1)\gamma_n^2 +m(b+\ell \gamma_n)^2 \right ) \|\dot \gamma \|^2
+ \frac{m}{2} \left ( J_n-J_1+m\ell^2 \right ) \dot \gamma_n^2 -  m\ell \mathcal{G}\gamma_n,
\end{equation}
where we have repeatedly used the condition $(\gamma, \dot \gamma )_{\R^n}=0$, that holds
in view of our realisation~\eqref{eq:TSn-1} of the tangent bundle $T\Ss^{n-1}$.

Apart from the $G=\SE(n-1)$ action that allows us to reduce the dynamics to $T^*\Ss^{n-1}$, 
the system possesses additional symmetries due to our  assumptions on the mass distribution of the sphere.
These correspond to rotations of the body frame that preserve the symmetry axis $E_n$. The symmetry group is
hence $A=\SO(n-1)$  and the action of $\tilde a\in \SO(n-1)$ on $(R,x)\in (\SO(n),x)$ is given by
\begin{equation}
\label{eq:H-action}
\tilde a\cdot (R,x) =(R  a^{-1},x), \quad  \mbox{where}  \quad  a:=\left ( \begin{array}{c|c} \tilde a  & 0 \\ \hline 0 & 1 \end{array}  \right ) \in \SO(n).
\end{equation}
The tangent lift of this action on $(R,\Omega, x, \dot x)\in \SO(n)\times \so(n)\times \R^n\times\R^n$ is
\begin{equation*}
\tilde a\cdot (R,\Omega, x, \dot x) =(R a^{-1},\Ad_{a} \Omega, x, \dot x), 
\end{equation*}
Using that $a^{-1}\J a=\J$ in view of Eq.~\eqref{eq:MassTensorGen}, 
one may check that the  Lagrangian~\eqref{eq:Lag-nD} is invariant. The same
is true about the constraints~\eqref{eq:nDrolling} and~\eqref{eq:nDrubber}  so the dynamics is $A$-equivariant. 

A crucial observation is that the $A$-action defined by Eq.~\eqref{eq:H-action} commutes with the $G$-action defined by Eq.~\eqref{eq:action}, so there is a well  defined $A$-action on the shape space $\Ss^{n-1}$. 
As may be
easily shown from Eq.~\eqref{eq:H-action}  and the definition of $\gamma =R^{-1}e_n$, such action is by rotations of the sphere
$\Ss^{n-1}$ that fix the vertical axis. Namely, with the same notation for $a$ and $\tilde a$ as above:
\begin{equation*}
\tilde a \cdot \gamma = a \gamma, \qquad \gamma \in \Ss^{n-1},
\end{equation*}
where we recall that $\Ss^{n-1}$ is realised by its embedding in $\R^n$~\eqref{eq:Sn-1}.  In particular, this action
fixes the north and south pole of $\Ss^{n-1}$ and therefore is non-free.

The tangent lift of this action to 
$T\Ss^{n-1}$ is $\tilde a \cdot ( \gamma,  \dot \gamma) = ( a \gamma , a\dot \gamma)$  
and it is immediate to check that  it leaves 
the reduced Lagrangian~\eqref{eq:red-Lag-Routh}  invariant. It is also 
clear that the function
 $\phi$ given
by~\eqref{eq:phi-nD} is $A$-invariant so the hypothesis to apply Theorem~\ref{Th:Noether} hold.

The Lie algebra $\frak{a}=\so(n-1)$ is naturally identified with the set of $n\times n$ skew-symmetric matrices 
$\xi\in \so(n)$ such that $\xi E_n=0$. The infinitesimal generator of $E_i \wedge E_j \in \frak{a}$, 
$1\leq i,j \leq n-1$, is the vector field on $\Ss^{n-1}$ given by
\begin{equation*}
(E_i\wedge E_j)_{\Ss^{n-1}}(\gamma)=(E_i\wedge E_j)\gamma \in T_\gamma \Ss^{n-1}.
\end{equation*}

Using the expression~\eqref{eq:red-Lag-Routh} for the reduced Lagrangian, we compute the
action of the rescaled tangent bundle momentum map $\mathcal{J}:T\Ss^{n-1}\to \so(n-1)^*$~defined by \eqref{eq:tangentbundlemommap} on 
$ E_i\wedge E_j\in \so(n-1)$ to be  given by
\begin{equation*}
\begin{split}
\mathcal{J}(\gamma, \dot \gamma) (E_i\wedge E_j) &= \exp(\phi) \left [  \left ( 2J_1 +(J_n-J_1)\gamma_n^2 +m(b+\ell \gamma_n)^2 \right )  \left ( (E_i\wedge E_j) \gamma, \dot \gamma \right )_{\R^n}  \right . \\
& \left . \qquad \qquad \qquad + m \left ( J_n-J_1+m\ell^2 \right ) \dot \gamma_n \left ( (E_i\wedge E_j) \gamma, E_n \right )_{\R^n}  \right ] \\
&= \sqrt{2J_1 +(J_n-J_1)\gamma_n^2 +m(b+\ell \gamma_n)^2 } \, ( \dot \gamma_i \gamma_j - \dot \gamma_j
\gamma_i),
\end{split}
\end{equation*}
where, in the second equality, we have used the specific form~\eqref{eq:phi-nD} of the function $\phi$.
The quantity  $\dot \gamma_i \gamma_j - \dot \gamma_j \gamma_i$ is the $i$-$j$ entry of the matrix 
$-\gamma \wedge \dot \gamma$, which by Eq.~\eqref{eq:Omega-as-gammawedgedotgamma} coincides
with $-\Omega$. Therefore, by Theorem~\ref{Th:Noether}, the functions $F_{ij}$ given
by~\eqref{eq:FirstIntegralsRouth} are first integrals of the system as claimed.

\appendix

\section{Proof of Lemma~\ref{l:MetricXi}.}
\label{app:proofs-lemmas}

\subsection{Proof of~\eqref{eq:Kij}.}
\label{app:lemmaA1}

The proof is a calculation for which we outline the details. Taking into account the form of the kinetic energy metric
of the Lagrangian~\eqref{eq:Lag-nD}, and the expressions~\eqref{eq:defXi} for the horizontal lifts $h_i$, it follows that  for $k,l \in \{1,\dots, n-1\}$ we may write
\begin{equation}
\label{eq:K_kl-aux}
\llangle h_k,h_l \rrangle =A_{kl}+B_{kl},
\end{equation}
where
\begin{equation*}
\begin{split}
A_{kl}:  &= \left ( \I \left (  \gamma\wedge \left ( E_k  -\frac{\gamma_k}{\gamma_n}E_n \right ) \right ),  \gamma\wedge \left ( E_l  -\frac{\gamma_l}{\gamma_n}E_n \right ) \right )_\kappa,
\\
B_{kl}: & = m\left (    a \left (- E_k  + \frac{\gamma_k}{\gamma_n}E_n \right )+ \ell   \gamma\wedge \left ( E_k  -\frac{\gamma_k}{\gamma_n}E_n \right )E_n ,  a \left (- E_l  + \frac{\gamma_l}{\gamma_n}E_n \right )+ \ell   \gamma\wedge \left ( E_l  -\frac{\gamma_l}{\gamma_n}E_n \right )E_n\right )_{\R^n},
\end{split}
\end{equation*}
with $(\cdot, \cdot)_{\R^n}$ denoting the euclidean norm in $\R^n$.

We first compute the value of $A_{kl}$. Using the expressions~\eqref{eq:InTensorGen} and
\eqref{eq:MassTensorGen} for the inertia tensor one verifies that
\begin{equation}
\label{eq:inertia-appendix}
 \I \left (  \gamma\wedge \left ( E_k  -\frac{\gamma_k}{\gamma_n}E_n \right ) \right )
 = 2J_1  \left (  \gamma +  \frac{(J_n-J_1)\gamma_n}{2J_1}  E_n  \right ) \wedge
  \left ( E_k  + \frac{(J_1+J_n)\gamma_k}{2J_1\gamma_n} E_n \right ), \qquad j=1,\dots, n-1.
\end{equation}
The above expression, together with the general identity 
\begin{equation}
\label{eq:killing-appendix}
(u_1\wedge v_1,u_2\wedge v_2)_\kappa= (u_1,u_2)_{\R^n}(v_1,v_2)_{\R^n}-(u_1,v_2)_{\R^n}(u_2,v_1)_{\R^n},
\end{equation}
that holds for $u_1,v_1,u_2,v_2\in \R^n$, leads to
\begin{equation}
\label{eq:Aklaux}
A_{kl}=(2J_1+(J_n-J_1)\gamma_n^2)\delta_{kl}+ 
\left ( \frac{J_1+J_n}{\gamma_n^2} +J_n-J_1 \right ) \gamma_k\gamma_l. 
\end{equation}

On the other hand we have
\begin{equation}
\label{eq:aux-Bkl}
\gamma \wedge \left ( E_j - \frac{\gamma_j}{\gamma_n} E_n \right ) E_n= \gamma_j E_n-\gamma_n E_j - 
\frac{\gamma_j}{\gamma_n} \gamma,
\end{equation}
so we may write
\begin{equation}
\begin{split}
\label{eq:Bklaux}
B_{kl}  &= m\left (  \left (\ell + \frac{a}{\gamma_n} \right ) \gamma_k E_n - \left ( a + \ell \gamma_n \right ) E_k
- \ell \frac{\gamma_j}{\gamma_n} \gamma \, ,  \, 
 \left (\ell + \frac{a}{\gamma_n} \right ) \gamma_l E_n - \left ( a + \ell \gamma_n \right ) E_l
- \ell  \frac{\gamma_l}{\gamma_n} \gamma \right )_{\R^n} \\
&= m(a + \ell \gamma_n)^2 \delta_{kl}+ m\left ( \left ( \frac{a}{\gamma_n} +\ell \right )^2+\frac{\ell^2}{ \gamma^2_n} \right )
 \gamma_k\gamma_l.
\end{split}
\end{equation}
Substitution of~\eqref{eq:Aklaux} and~\eqref{eq:Bklaux} into~\eqref{eq:K_kl-aux}
proves~\eqref{eq:Kij}.

\subsection{Proof of~\eqref{eq:XiXjXl}.}
\label{app:lemmaA2}
The crucial part of the proof is to obtain the following expression for the Jacobi-Lie bracket of the vector fields 
$h_i$ and $h_j$:
\begin{equation}
\label{eq:commutator}
[h_i,h_j](R,x)=  \left ( \left ( E_i  -\frac{\gamma_i}{\gamma_n}E_n \right ) \wedge  \left ( E_j  -\frac{\gamma_j}{\gamma_n}E_n \right ) \, , 
\, 0 \, \right )\in \so(n)\times \R^n, \qquad i,j=1,\dots, n-1.
\end{equation}

Accept that this is the case for the moment. Then, 
similar to the calculation performed in section~\ref{app:lemmaA1},
we have
\begin{equation}
\label{eq:aux-lemma-jkl}
\llangle [h_i,h_j] , h_l \rrangle  = \tilde A_{ijl}+ \tilde B_{ijl},
\end{equation}
where
\begin{equation*}
\tilde A_{ijl} =  \left ( \I \left (  \gamma\wedge \left ( E_l  -\frac{\gamma_l}{\gamma_n}E_n \right ) \right ),  
  \left ( E_i  -\frac{\gamma_i}{\gamma_n}E_n \right ) \wedge  \left ( E_j  -\frac{\gamma_j}{\gamma_n}E_n \right ) 
  \right )_\kappa ,
\end{equation*}
and 
\begin{equation*}
 \tilde B_{ijl} =m  \left (\ell \left ( E_i  -\frac{\gamma_i}{\gamma_n}E_n \right ) \wedge  \left ( E_j  -\frac{\gamma_j}{\gamma_n}E_n  \right ) E_n\, , \,  
 a \left (- E_l  + \frac{\gamma_l}{\gamma_n}E_n \right )+ \ell   \gamma\wedge \left ( E_l  -\frac{\gamma_l}{\gamma_n}E_n \right ) E_n
  \right )_{\R^n}.
\end{equation*}
On the one hand, using again~\eqref{eq:inertia-appendix} and~\eqref{eq:killing-appendix}, one may
simplify
\begin{equation}
\label{eq:Aijl-tilde}
\tilde A_{ijl} = (J_n-J_1)(\gamma_j\delta_{il} -\gamma_i \delta_{jl}).
\end{equation}
On the other hand, using 
\begin{equation*}
 \left ( E_i  -\frac{\gamma_i}{\gamma_n}E_n \right ) \wedge  \left ( E_j  -\frac{\gamma_j}{\gamma_n}E_n 
  \right ) E_n
 =\frac{1}{\gamma_n}(\gamma_i E_j - \gamma_j E_i ),
\end{equation*}
together with~\eqref{eq:aux-Bkl}, allows one to write
\begin{equation*}
\begin{split}
\tilde B_{ijl} &=\frac{m\ell}{\gamma_n} \left ( \gamma_i E_j - \gamma_j E_i \, , \,  \left (\ell + \frac{a}{\gamma_n} \right ) \gamma_l E_n - \left ( a + \ell \gamma_n \right ) E_l
- \ell  \frac{\gamma_l}{\gamma_n} \gamma \right )_{\R^n} \\
& = m\ell \left ( \ell + \frac{a}{\gamma_n} \right ) \left ( \gamma_j \delta_{il}-\gamma_i \delta_{lj} \right ).
\end{split}
\end{equation*}
Substitution of the above expression, together with~\eqref{eq:Aijl-tilde}, onto~\eqref{eq:aux-lemma-jkl} 
proves~\eqref{eq:XiXjXl}.

Hence, to complete the proof, it only remains to establish the validity of~\eqref{eq:commutator}. The formula clearly holds
 for $i=j$, so below  we assume that $i\neq j$.
 As a consequence of the independence  of $h_i$ on $x$,  we have
\begin{equation*}
[h_i,h_j](R,x)=  \left ( \,  \left [  \gamma\wedge \left ( E_i  -\frac{\gamma_i}{\gamma_n}E_n \right )  
\, , \, \gamma\wedge \left ( E_j  -\frac{\gamma_j}{\gamma_n}E_n \right )  \right ]_{\SO(n)}
\, , \, W_{ij} \right )\in \so(n)\times \R^n, 
\end{equation*}
where $[\cdot , \cdot ]_{\SO(n)}$ is the Lie bracket of vector fields on $\SO(n)$ (written  in the 
left trivialisation as usual) and $W_{ij}\in \R^n$ has components
\begin{equation}
\label{eq:Wij}
W_{ij}^{(k)}= a \gamma\wedge \left ( E_i  -\frac{\gamma_i}{\gamma_n}E_n \right ) \left  [  -R_{kj} +
\frac{\gamma_j}{\gamma_n}R_{kn} \right ] -
 a\gamma\wedge \left ( E_j -\frac{\gamma_j}{\gamma_n}E_n \right )
  \left  [  -R_{ki} +
\frac{\gamma_i}{\gamma_n}R_{kn} \right ], \quad k=1,\dots, n-1, 
\end{equation}
$W_{ij}^{(n)}=0$, where $R_{kl}$ denotes the $k$-$l$ entry of the matrix $R\in \SO(n)$. 

On the one hand,  Garc\'ia-Naranjo and Marrero~\cite[Lemma 4.4]{GNMarr2018} compute:
\begin{equation*}
 \left [  \gamma\wedge \left ( E_i  -\frac{\gamma_i}{\gamma_n}E_n \right )  
\, , \, \gamma\wedge \left ( E_j  -\frac{\gamma_j}{\gamma_n}E_n \right )  \right ]_{\SO(n)}=
 \left ( E_i  -\frac{\gamma_i}{\gamma_n}E_n \right ) \wedge  \left ( E_j  -\frac{\gamma_j}{\gamma_n}E_n \right ),
\end{equation*}
which establishes the correctness of  the first entry of~\eqref{eq:commutator}. 

On the other hand, we shall prove 
 that 
\begin{equation}
\label{eq:Lemma-App-aux2}
 \gamma\wedge \left ( E_i  -\frac{\gamma_i}{\gamma_n}E_n \right ) \left  [  -R_{kj} +
\frac{\gamma_j}{\gamma_n}R_{kn} \right ] = \frac{R_{kn}}{\gamma_n^3}\gamma_i\gamma_j, \qquad k=1,\dots, n-1.
\end{equation}
Considering that a similar formula holds  when the roles of $i$ and $j$ are interchanged, it follows from
\eqref{eq:Wij} that
$W_{ij}$ vanishes and Eq.~\eqref{eq:commutator} indeed holds. The calculations to establish~\eqref{eq:Lemma-App-aux2} 
rely on the following identity whose proof may be found in Garc\'ia-Naranjo and Marrero~\cite[Lemma B.1]{GNMarr2018}:
\begin{equation}
\label{eq:basic-appendix}
E_i\wedge E_j\, [R_{kl}] =R_{ki}\delta_{jl} - R_{kj}\delta_{il}, \qquad i,j,k,l\in \{1, \dots, n\}.
\end{equation}
Using~\eqref{eq:basic-appendix}, and writing $\gamma=\sum_{l=1}^n\gamma_l E_l$ and $\gamma_l=R_{nl}$,
 it is straightforward to obtain
 (recall that we assume that $i,j,k \in \{1, \dots, n-1\}$ and $i\neq j$):
\begin{equation}
\label{eq:Lemma-App-aux3} 
\gamma \wedge E_i \left [ R_{kj} \right ] = -\gamma_j R_{ki}=\gamma \wedge E_i \left [  \frac{\gamma_j}{\gamma_n} R_{kn} \right ],
\qquad \gamma \wedge E_n \left [ R_{kj} \right ]= - \gamma_jR_{kn}.
\end{equation}
With a little bit more work, and using $\sum_{l=1}^nR_{kl}\gamma_l=\sum_{l=1}^nR_{kl}R_{nl}=\delta_{kn}=0$, one obtains
\begin{equation}
\label{eq:Lemma-App-aux4} 
\gamma \wedge E_n  \left [  \frac{\gamma_j}{\gamma_n} R_{kn} \right ] = -\frac{\gamma_j}{\gamma_n^2}R_{kn} -\gamma_j R_{kn}.
\end{equation}
Identities~\eqref{eq:Lemma-App-aux3} and \eqref{eq:Lemma-App-aux4} 
imply that~\eqref{eq:Lemma-App-aux2} holds.

\vspace{1cm}

\noindent \textbf{Acknowledgements:} The author acknowledges the Alexander von Humboldt Foundation  for a  Georg Forster Experienced Researcher Fellowship     that funded a research visit to TU Berlin where this work was done. He is also thankful to C. 
Fern\'andez for her help to produce Figure~\ref{F:Routh-sphere}. Finally he acknowledges Professor
V. Dragovi\'c for the invitation to submit this paper to the special issue in Theoretic and Applied Mechanics
in honor
of the 150th birthday of S. A. Chaplygin.

\vskip 1cm

\noindent LGN: Departamento de Matem\'aticas y Mec\'anica, 
IIMAS-UNAM.
Apdo. Postal 20-126, Col. San \'Angel,
Mexico City, 01000,  Mexico. luis@mym.iimas.unam.mx.  

\begin{thebibliography}{99}
\let\\, \newcommand{\by}[1]{\textsc{\ignorespaces #1}\\}
  \newcommand{\title}[1]{\textsl{\ignorespaces #1}\\}
  \newcommand{\vol}[1]{{\bf{\ignorespaces #1}}}
  \newcommand{\info}[1]{\textrm{\ignorespaces #1}.}

\small  \setlength{\parskip}{0pt}
  

  
%
%



%
%
%

%


\bibitem{BalseiroSansonetto}  \by{Balseiro, P.\  and N.~Sansonetto} A  geometric characterization of certain first integrals for nonholonomic systems with symmetries. {\em SIGMA Symmetry Integrability Geom. Methods Appl.} \vol{12} (2016),  14 pp. 







%
%
%
%



\bibitem{BKMM} \by{Bloch, A.M.,  P.S.~Krishnaprasad, J.E.~Marsden and R.M.~Murray} 
Nonholonomic mechanical systems with symmetry. \emph{Arch. Ration. Mech. Anal.} \vol{136} (1996), 21--99.


%


%

 
 
 
 
  \bibitem{BolsBorMam2015}  \by{ Bolsinov, A. V.,  A. V. Borisov and  I. S. Mamaev} Geometrisation of Chaplygin's
 Reducing Multiplier theorem \emph{Nonlinearity}, \vol{28}  (2015),  2307--2318.
 
 \bibitem{BorMamChap} \by{Borisov A.~V. and   I.~S. Mamaev} Chaplygin's Ball Rolling Problem Is Hamiltonian. {\em Math. Notes}, (2001),
{\bf 70}, 793--795. 
 
  \bibitem{BorMam2008} \by{Borisov, A. V. and Mamaev, I. S.} Conservation Laws, Hierarchy of Dynamics and Explicit Integration
of Nonholonomic Systems. \emph{Regul. Chaotic Dyn.} \vol{13} (2008),  443--490.


 

 \bibitem{BorMamHam}  \by{Borisov   A.V. and   I.S.~Mamaev}  Isomorphism and Hamilton Representation of Some Non-holonomic Systems, {\em Siberian Math. J.}, {\bf 48} (2007),  33--45 See also: arXiv: nlin.-SI/0509036 v. 1 (Sept. 21, 2005).
 
 
 
   \bibitem{BorMamRubber}  \by{Borisov   A.V., Mamaev I.S. and I.A. Bizyaev} The hierarchy of dynamics of a rigid body
 rolling without slipping and spinning on plane and a sphere. \emph{Regul.\ Chaotic Dyn.} \vol{18} (2013), 277--328.
 

 
 
%
%
%
%
 

%
%
%
%
%









 \bibitem{CaCoLeMa} \by{Cantrijn F, Cort\'es J., de Le\'on M. and D. Mart\'in de Diego} On the geometry of generalized Chaplygin systems. \emph{Math. Proc.
 Cambridge Philos. Soc.} \vol{132} (2002),  323--351.

%
%



\bibitem{ChapRedMult} \by{Chaplygin, S.A.} On the theory of the motion of nonholonomic systems. The Reducing-Multiplier Theorem. \emph{Regul. Chaotic Dyn.} \vol{13}, 369--376 (2008) [Translated from \textit{Matematicheski\v{i} Sbornik} (Russian) \vol{28} (1911), by A. V. Getling]





%
%


\bibitem{EhlersKoiller} \by{Ehlers, K., J. Koiller, R. Montgomery  and P.M. Rios}
 Nonholonomic  Systems via Moving Frames: Cartan Equivalence and Chaplygin
Hamiltonization.  in \emph{The breath of Symplectic and Poisson Geometry},
Progress in Mathematics Vol.\ \vol{232} (2004), 75--120.


%
%


%
%
%
%

%



%

\bibitem{Fasso2} \by{Fass\`o, F., A.~Giacobbe and N.~Sansonetto}  
Gauge conservation laws and the momentum equation in nonholonomic mechanics. \emph{Rep.\ Math.\ Phys.} \vol{62} (2008), 345--367.


\bibitem{FassoJGM} \by{Fass\`o, F.  and N.~ Sansonetto}  
An elemental overview of the nonholonomic Noether theorem, \emph{ Int.\ J.\ Geom.\ Methods Mod.\ Phys.} \vol{6} (2009), 1343--1355.


%
%


\bibitem{FGNS18}  \by{Fass\`o, F.,  Garc{\'ia}-Naranjo L. C., and  N.~Sansonetto} Moving energies as first integrals of nonholonomic systems with affine constraints. {\em Nonlinearity} \vol{31} (2018),  755--782.

\bibitem{FGNM18} \by{Fass\`o, F.,  Garc{\'ia}-Naranjo L. C., and  J. Montaldi} Integrability and dynamics of the $n$-dimensional Veselova top.   {\em J. Nonlinear Sci.} (2018). https://doi.org/10.1007/s00332-018-9515-5.





\bibitem{FedKoz} \by{Fedorov, Y. N., and V.V.~Kozlov} Various aspects of $n$-dimensional rigid body dynamics.
\emph{Amer.\ Math.\ Soc.\ Transl.\ (2)} \vol{168} (1995), 141--171.

\bibitem{FedJov} \by{Fedorov, Y. N. and B. Jovanovi\'c} Nonholonomic LR systems as generalized Chaplygin systems with an invariant measure and flows on homogeneous spaces. 
\emph{J.\ Nonlinear Sci.} \vol{14} (2004), 341--381. 

\bibitem{FedJov2} \by{Fedorov, Y. N. and B. Jovanovi\'c} Hamiltonization of the generalized Veselova LR system. 
\emph{Regul.\ Chaot.\ Dyn.} \vol{14} (2009), 495--505.


\bibitem{Fernandez} \by{Fernandez, O., Mestdag, T. and A.M. Bloch} A generalization of Chaplygin's reducibility Theorem.
{\em Regul. Chaotic Dyn.} \vol{14} (2009) 635--655. 

%
%




\bibitem{Gajic} \by{Gaji\'c~B. and B.~Jovanovi\'c.} Nonholonomic connections, time reparametrizations,
and integrability of the rolling ball over a sphere.  \href{https://arxiv.org/abs/1805.10610}{arXiv: 1805.10610}  (2018)
 
%
%




\bibitem{LGN-JM17}  \by{Garc\'ia-Naranjo, L.C. and J.~Montaldi} 
Gauge momenta as Casimir functions of nonholonomic systems. \emph{Arch. Ration. Mech. Anal.} \textbf{228} (2018), 563--602.  

\bibitem{LGN18}  \by{Garc\'ia-Naranjo, L.C.} 
Generalisation of Chaplygin's Reducing Multiplier  Theorem with an application to multi-dimensional nonholonomic
dynamics.  \href{https://arxiv.org/abs/1805.06393}{arXiv: 1805:06393}  (2018)


\bibitem{GNMarr2018}  \by{Garc\'ia-Naranjo, L.C. and J.C.~Marrero} 
The geometry of nonholonomic Chaplygin systems revisited.   \href{https://arxiv.org/abs/1812.01422}{arXiv: 1812.01422}  (2018)

%
%
%
%
%












\bibitem{HochGN} \by{Hochgerner~S. and L.~C.Garc\'ia-Naranjo}  $G$-Chaplygin systems with internal symmetries, truncation, and an (almost) symplectic view of Chaplygin's ball. {\em J. Geom. Mech.} {\bf 1} (2009), 35--53.


%
%

%

\bibitem{Iliev1975} \by{Iliev, I.} On first integrals of a nonholonomic mechanical system
\emph{J. Appl.
Math. Mech.} \vol{39} (1975), 147--150.



\bibitem{Iliev1985} \by{Iliev, I.}  On the conditions for the existence of the reducing Chaplygin factor. \emph{J. Appl.
Math. Mech.} \vol{49} (1985), 295--301.

\bibitem{Jov2003}  \by{Jovanovi\'c, B.}  Some multidimensional integrable cases of nonholonomic rigid body dynamics. {\em 
Regul. Chaotic Dyn.} \vol{8} (2003),  125--132.

\bibitem{JovaRubber} \by{Jovanovi\'c, B.}  LR and L+R systems. \emph{J. Phys. A}  {\bf 42} (2009),  18 pp.

\bibitem{Jovan} \by{Jovanovi\'c, B.} Hamiltonization and integrability of the Chaplygin sphere in $\R^n$. \emph{J. Nonlinear Sci.} \vol{20} (2010), 569--593.

\bibitem{Jova18} \by{Jovanovi\'c, B.} Rolling balls over spheres in $\R^n$. \emph{Nonlinearity}, \vol{31}  (2018),  4006--4031.

%
\bibitem{Koi} \by{Koiller, J.} 
 Reduction of some classical nonholonomic systems with symmetry.
\emph{Arch. Ration. Mech. Anal.} \vol{118} (1992), 113--148.

\bibitem{KoillerRubber} \by{Koiller J. and K. Ehlers} Rubber rolling over a sphere. \emph{Regul.\ Chaot.\ Dyn.} \vol{12} (2006), 127--152.







\bibitem{LeMa} \by{de Le\'on, M. and Mart{\'\i}n de Diego, D.} On the geometry of nonholonomic Lagrangian systems. \emph{J. Math. Phys.} \vol{
37} (1996), 3389--3414
%
%
%

\bibitem{MaRa} \by{Marsden J.E. and T.S.~Ratiu} \title{Introduction to Mechanics with Symmetry} 
\info{Texts in Applied Mathematics  \vol{17} Springer-Verlag 1994}



%

%
%

\bibitem{Ratiu80} \by{Ratiu, T.S.} The motion of the free n-dimensional rigid body. \emph{Indiana Univ. Math. J.} \vol{29} (1980),  609--629.







\bibitem{Routh} \by{Routh, E.D.} 
\title{Dynamics of a system of rigid bodies}  \info{7th ed., revised and enlarged. Dover Publications, Inc., New York, 1960}


\bibitem{Stanchenko} \by{Stanchenko, S.} Nonholonomic Chaplygin systems. \emph{Prikl. Mat. Mekh.} \vol{53},16--23;
English trans.: \emph{J. Appl. Math. Mech.} \vol{53} (1989),  11--17.



%
%
%

%
%


\bibitem{ZenkovSuslov}  \by{Zenkov~D.V.\ and Bloch~A.M.} Dynamics of the $n$-dimensional Suslov problem {\em J. Geom. Phys.} \vol{34} (2000), 121--136.



\end{thebibliography}
\end{document}